\documentclass[11pt,a4paper,english]{article} 
\usepackage{authblk}
\title{A Deformation Quantization for  Non-Flat Spacetimes and   Applications to QFT}
\author[1]{Albert Much\footnote{much@itp.uni-leipzig.de}}  
\affil[1]{Institut f\"ur Theoretische Physik\\ Universit\"at Leipzig\\ D-04103 Leipzig}  
\usepackage{float}            
\usepackage[T1]{fontenc} 
\usepackage[cmex10]{amsmath}  
\usepackage{amstext}          
\usepackage{mathrsfs}
\usepackage{amsfonts}         
\usepackage{amssymb}          
\usepackage{amsthm}    
\usepackage{setspace}
\usepackage{bm}               
\usepackage{enumerate}        
\usepackage[bottom]{footmisc} 
\usepackage{array}            
\usepackage{textcomp}         
\usepackage{pdfpages}         
\usepackage{parskip}          
\usepackage[right]{eurosym}   
\usepackage{xcolor}           
\usepackage[square,numbers]{natbib}	  
\usepackage{tikz}
\usetikzlibrary{matrix}
\usepackage{bbold}
\usepackage{amsmath}  
\usepackage[colorlinks,pdfpagelabels,pdfstartview = FitH,bookmarksopen = true,bookmarksnumbered =
true,linkcolor = black,plainpages = false,hypertexnames = false,citecolor = black]{hyperref}
\usepackage[a4paper,textwidth=17cm,textheight=24cm,headheight=14pt]{geometry}
\allowdisplaybreaks

\renewcommand{\vec}[1]{{\ifnum9<1#1\mathbf{#1}\else\ifcat\noexpand#1\relax\boldsymbol{#1}\else\mathbfi{#1}\fi\fi}}
\newcommand{\mathe}{\mathrm{e}}
\newcommand{\mathi}{\mathrm{i}}
\let\oldre\Re
\let\oldim\Im
\renewcommand{\Re}{\oldre\mathfrak{e}\,}
\renewcommand{\Im}{\oldim\mathfrak{m}\,}
\newcommand{\total}{\mathop{}\!\mathrm{d}}

\newcommand{\abs}[1]{{\left\lvert{#1}\right\rvert}}
\newcommand{\norm}[1]{{\left\lVert{#1}\right\rVert}}

\newcommand{\eqend}[1]{\,#1}

\newcommand{\ma}{\mathcal{M}}

\newcommand{\tx}{T_X}
\newcommand{\gm}{\mathcal{G}_{\mathcal{M}}}

\newcommand{\pz}{\text{Proj}_Z} 
\newcommand{\pza}{\text{Proj}_{Z_1}} 
 
\newtheorem{theorem}{\textsc{Theorem}}[section]
\newtheorem{lemma}[theorem]{\textsc{Lemma}}
\newtheorem{proposition}[theorem]{\textsc{Proposition}}

\newtheorem{corollary}[theorem]{\textsc{}Corollary}
\newtheorem{definition}[theorem]{\textsc{Definition}}

\newtheorem{remark}[theorem]{Remark}
\newtheorem{example}[theorem]{Example}

\newcommand{\R}{\mathbb{R}}
\newcommand{\C}{\mathbb{C}}

\newcommand{\N}{\mathbb{N}}

\usepackage{fancyhdr}
\numberwithin{equation}{section} 
\begin{document}
	\maketitle
	\abstract{We provide a deformation quantization, in the sense of Rieffel, for \textit{all} globally hyperbolic spacetimes with a Poisson structure. The Poisson structures have to satisfy Fedosov type requirements in order for the deformed product to be associative. We apply the novel deformation to quantum field theories and their respective states and we prove that the deformed state (i.e.\ a state in non-commutative spacetime) has a singularity structure resembling Minkowski, i.e.\ is \textit{Hadamard}, if the undeformed state is Hadamard. This proves that the Hadamard condition, and hence the quantum field theoretical implementation of the equivalence principle is a general concept that holds in spacetimes with quantum features (i.e. a non-commutative spacetime).}
	\tableofcontents
	
	\section{Introduction}
	Quantum field theory (\textbf{QFT}) in curved spacetimes (\textbf{CST}) (see monographs \cite{FU, WQ, CF, BD, GBP} and \cite{K88}) is regarded as a zero-order approximation of a complete theory of quantum gravity. The approximation pertains to the treatment of gravity (or spacetime) as a classical entity rather than a quantum one, as in general relativity.
	
	Even in this approximation, the framework has already provided various insights and achievements, such as thermal radiations by black holes (\cite{IN1,WPC,FH1,K88}), the micro-local spectrum condition (\cite{RAD, RAD2, IN3}), Pauli's spin-statistic theorem in curved spacetime (\cite{IN4}), quantum inequalities (\cite{FRP,FRP2, IN5,IN6,IN7,IN9, IN8}), perturbation theory (\cite{INN1,INN2,INN3,INN5}), descriptions of the early universe and extreme astrophysical environments \cite{THCOS} (and references therein).
	
	Moreover, it has led to questioning the foundations of physics, such as the very notion (or existence) of particles (Fulling-Davies-Unruh effect, \cite{INN7,INN4,INN6,DPE}).
	
	The states generalizing the Minkowski vacuum in QFT in CST are the so-called Hadamard states. These states exhibit the Hadamard singularity structure and exist for a large class of quantum field theories in globally hyperbolic spacetimes, \cite{HS1, SV01}, see also \cite{GW12}. An essential reason to work with them is the lack of a vacuum (lack of an isometry group\footnote{In flat spacetime, this is the Poincar\'e group, which leaves the vacuum invariant.}) in curved spacetimes. Motivated by the equivalence principle, one demands states resembling the same singularity structure as the Minkowski vacuum at short distances.
	
	Due to the foundational importance of QFT in CST and the zero-order approximation argument, it is clear that any fundamental and complete theory of quantum gravity should at least contain QFT in CST in some limit. In addition, this theory should become important where the domain of validity of QFT in CST breaks down. This breakdown occurs at distances (or curvatures) near the Planck scales, \cite[Chapter 1]{WQ}.
	
	It is exactly around those distances that the hypothesis of a quantum (or noncommutative) spacetime (geometry) comes into play as the solution to the geometrical measurement problem, see Refs.\cite{Ahluwalia1993,DFR}. The problem arises due to limitations of localization of spacetime points, resulting from combining the quantum-mechanical uncertainty principle with the formation of black holes in general relativity. An $n$-dimensional quantum spacetime is defined as an algebra $\mathcal{V}$, generated by elements $\hat{x}$, that obey
	\begin{equation}
		\label{eq:commink}
		\left[ \hat{x}^\mu, \hat{x}^\nu \right] = 2 \mathi \Theta^{\mu\nu},
	\end{equation}
	where $\Theta$ is a real, constant, skew symmetric $n\times n$ matrix. By representing the generating elements $\hat{x}$ on a Hilbert space, one derives an uncertainty principle for the coordinates of spacetime. In \cite{DFR}, the authors defined a quantum field theory in noncommutative Minkowski space, by defining the representation space of the QFT to be the tensor product of $\mathcal{V} \otimes \mathscr{H}$, where $\mathscr{H}$ is the Hilbert space of the quantum field under consideration. Next, the authors in Ref. \cite{GL1} constructed a unitary operator that maps the tensor product space $\mathcal{V} \otimes \mathscr{H}$ to the Hilbert space $\mathscr{H}$, which is a strict deformation of QFT called warped convolutions~\cite{BS,BLS}.
	
	This prescription is connected to the Rieffel product~\cite{R} (also known as the Moyal-Weyl-Groenewold), which is a noncommutative associative product of functions provided by
	\begin{equation}
		\label{eq:rfp}
		f \star_\theta g = (2\pi)^{-n}\iint \alpha_{\Theta x}(f) \alpha_y(g) \mathe^{\mathi x \cdot y} \total^n x \total^n y \eqend{,}
	\end{equation}
	where $\alpha$ denotes the action of the translation group $\R^n$, i.e.\ $\alpha_a (\phi)(x) = \phi(x+a)$.
	The oscillatory integral is well defined for smooth functions of compact support $f,g\in C_0^{\infty}(\mathbb{R}^n)$. To introduce perturbations stemming from the \textit{non-commutative spacetime}, one replaces the pointwise product by the Rieffel product. This procedure is equal to defining the $n$-point functions for a QFT in non-commutative Minkowski space, as done in \cite{OECK1,OECK2}, using the fields in \cite{GL1}.
	
	The Rieffel deformation, yet fruitful, is generally inapplicable to curved spacetimes due to the essential requirement of smooth actions of $\mathbb{R}^n$, i.e., mutually commutative generators. In rare cases, one can use commuting Lie vector fields~\cite{Mor} or a spacetime with commuting Killing vector fields~\cite{DFM1}.
	
	Recently, this limitation was circumvented for the De Sitter spacetime by using the embedding formalism. The idea was to embed an $n$-dimensional curved spacetime manifold into a higher dimensional Minkowski spacetime, thus supplying us with an $n$-dimensional group of translations with which we define the Rieffel product, \cite{MF}.
	
	In what follows, we first use \cite[Theorem 1.1]{MS10}, which provides an isometric embedding for all globally hyperbolic spacetimes into a higher dimensional Minkowski spacetime $\mathbb{L}^N$. Then, we use those embeddings together with recent techniques from optimization problems on smooth manifolds, see \cite{mahony, boumaye}, to define the equivalent of translations (including the exponential map), namely retractions on those manifolds. The retractions form a group, and using this group structure, we define a generalized Rieffel product for curved manifolds. This product is parameterized by a deformation parameter, and it reduces to the well-established Rieffel product in the flat case. Furthermore, at   first order in the deformation parameter, the generalized deformed product aligns with the Poisson bracket. We prove that under a certain assumption, namely a Fedosov type requirement on the Poisson structure $\pi$, i.e.\ $\nabla \pi=0$, that the product is associative up to second order in the deformation parameter.
	
	Since our main motivation to study these deformations lies in the application to QFT in curved spacetimes, and in particular to the two-point function, we generalize the deformed product to two points.
	
	In the final section, we demonstrate that if a state satisfies the microlocal spectrum condition—implying, according to the Radzikowski theorem, that it is Hadamard (possessing the same singularity structure as in Minkowski spacetime), then the deformed state also satisfies this condition. Consequently, we conclude that the deformed state is Hadamard. Moreover, beyond the aforementioned results, we establish that the Hadamard property is independent of the chosen isometric embedding and/or retraction.
	
	\section{Mathematical Preliminaries} 
	We start with basic definitions  needed for the following results.
	\subsection{Embedding Formalism}
	In Ref.~\cite{MS10} the authors proved the existence of an isometric embedding of any globally hyperbolic manifold in to a higher dimensional  Minkowski space $\mathbb{L}^N$. First, we define a smooth embedded submanifold of a linear space $\mathcal{E}$ of dimension $N$, \cite[Defintion 3.10.]{boumaye}.\footnote{For example the  Minkowski space $\mathbb{L}^N$ and the Euclidean space $\R^n$  fall into this category.} 
	\begin{definition}\label{def1}(\textbf{Embedded submanifold})
		Let $\mathcal{E}$ be a linear space of dimension $N$. A non-empty subset $\mathcal{M}$ of $\mathcal{E}$ is a (smooth) embedded submanifold of $\mathcal{E}$ of dimension $d=N-k$ if for some $k\geq1$ and, for each $X\in\mathcal{M}$, there exists a neighbourhood $U$ of $X$ in $\mathcal{E}$ and a smooth function $F:U\rightarrow \R^k$ such that $\,$ \newline
		\begin{enumerate}
			\item if $X$ is in $U$, then $F(X)=0$ if and only if $X\in \mathcal{M}$; and  $\,$ \newline 
			\item $\text{rank}\,\,DF(X)=k$.$\,$ \newline 
		\end{enumerate}
		Such a function $F$ is called a \textbf{local defining function} for $\mathcal{M}$ at $X$.
	\end{definition}

	The Nash embedding theorem states that every Riemannian manifold (compact or not) can be isometrically embedded in $\mathbb{R}^N$ for some large enough $N$. The authors in Ref.~\cite{MS10}   extended this result to Lorentzian manifolds that are stably causal with a steep temporal function, i.e.\ a smooth function $\tau$ such that $g(\nabla\tau,\nabla\tau)\leq-1$. In particular, one can prove that \textit{any} globally hyperbolic spacetime $(\mathcal{M},g)$ admits a steep temporal function $\tau$ and thus a Cauchy orthogonal decomposition, i.e.\
	\begin{align}\label{eqghst}
		\mathcal{M}=\R\times\Sigma,\qquad g=-\beta d\tau^2+g_{\tau},
	\end{align}
	where $\beta>0$ is a function on $\mathcal{M}$, $g_{\tau}$ is a Riemannian metric on $\Sigma_{\tau}:=\{\tau\}\times\Sigma$ smoothly varying with $\tau$, with upper bounded function $\beta$.
	And it is exactly in those spacetimes in which quantum field theories are rigorously defined.  The theorem that establishes the embedding is the following, \cite[Theorem 1.1]{MS10}.
	
	\begin{theorem}\label{thm:embsm}
		Let $(\mathcal{M},g)$ be a  Lorentzian manifold. The
		following assertions are equivalent:$\,$\\
		\begin{enumerate}
			\item $(\mathcal{M},g)$ admits a  isometric embedding in $\mathbb{L}^N$ for some $N\in \N$.$\,$\\
			\item  $(\mathcal{M},g)$ is a stably causal spacetime with a {\em steep
				temporal function}, i.e., a smooth function $\tau$ such that 
			$g(\nabla \tau, \nabla \tau)\leq -1$. 
		\end{enumerate}
		
	\end{theorem}
	
	\begin{definition}\label{def:ex}
		Let $\iota_{M,g}$ denote the isometric embedding of a globally hyperbolic manifold $(\mathcal{M},g)$ in   $\mathbb{L}^N$. 
	\end{definition} 
	
	\subsection{Moving  in the Embedded Spacetime} For the following considerations we need the concept of tangent spaces w.r.t\ points on the embedded manifolds, \cite[Theorem 3.15.]{boumaye}.
	\begin{theorem}
		Let $\mathcal{M}$ be an embedded submanifold of $\mathcal{E}$. Then, the tangent space $T_X\mathcal{M}$ is given as the kernel of the differential of any local defining function $F$ at $X$,
		\begin{equation}
			T_X\mathcal{M}=\text{ker}\, DF(X).
		\end{equation}
	\end{theorem}

	We endow the  embedded manifolds with a metric by using the fact that each tangent space $\tx\ma$ is a linear subspace of $(\mathbb{R}^N,\eta)$ or $(\R^N,\delta)$. Hence, we can define the metric product (not necessarily positive definite) $$(\cdot,\cdot)_X:\tx\ma\times\tx\ma\rightarrow\R,$$  by the restriction of the metric product in $\mathbb{L}^N$ or $\R^N$ to $\tx\ma$. 
	Next, we define the concept of the projector, see \cite[Definition 3.60.]{boumaye}. The idea of a projector is to define gradients or tangent vectors, by natural definitions in the embedding manifold  and projecting those quantities back to the embedded manifold.  
	\begin{definition}[\textbf{Orthogonal Projector}]\label{defproj}
		Let $\mathcal{M}$ be an embedded submanifold of a (Pseudo-)Euclidean space $\mathcal{E}$ equipped with a (Pseudo-)Euclidean metric $(\cdot,\cdot)$. The orthogonal projector to $\tx\ma$ is the linear map $\text{Proj}_X:\mathcal{E}\rightarrow\mathcal{E}$ characterized by the following properties: $\,$\newline
		\begin{enumerate}
			\item Range: ${im}({Proj}_X)=T_X\ma$;$\,$\newline
			\item Projector: ${Proj}_X\circ{Proj}_X={Proj}_X$;$\,$\newline
			\item  Orthogonal $(u-Proj_X(u),v)=0$ for all $v\in\tx\ma$ \label{it3} and $u\in\mathcal{E}$.
		\end{enumerate}
	\end{definition}
	\begin{lemma}\label{lem:orthproj}
		The orthogonal projector, which is a smooth map \cite[Lemma 3.1.2]{salamon}, can be represented by a matrix $P\in\R^{N\times N}$ 
		\begin{align*}
			\pz(V^A)=P^{A}_{\,\,C}(Z)\,V^C.
		\end{align*}
		which is uniquely determined by the conditions (see \cite[Proposition 3.63]{boumaye}, \cite[Chapter 3, Equation 3.1.2, Equation 3.1.3]{salamon})
		\begin{align*}
			P(Z)=P(Z)^2= P(Z)^T,
		\end{align*} 
		for $Z\in \ma$ and $V \in T_Z\ma$,
		\begin{align*}
			P (Z)V=V.
		\end{align*} 
		
	\end{lemma}
	
	Using the orthogonal projector we have the following (extended) result, see \cite[Proposition 3.61., Equation (3.36)]{boumaye}.
	\begin{proposition}\label{propgrad}
		Let $(\ma,g)$ be a (Pseudo-)Riemannian submanifold of $(\mathcal{E}, (\cdot,\cdot)))$  and let
		$f:\ma\rightarrow\R$ be a smooth function. The (Pseudo-)Riemannian gradient
		of $f$ is defined by
		\begin{align*}
			\nabla_A f(X)=\text{Proj}_X (\partial_A \bar{f}(X)),
		\end{align*} 
		where $ \bar{f}$ is any smooth extension of $f$ to a neighborhood of $\ma$ in $\mathcal{E}$ and $\partial$ denotes the  (Pseudo-)Euclidean gradient w.r.t.\ the coordinate $X$. Moreover, $\forall(X,V)\in T\mathcal{M}$
		\begin{align}
			(V,\nabla f )_{X}=    (V,\partial \bar{f} ).
		\end{align} 
		
	\end{proposition}  
	Thus, the scalar product with a tangent vector induces the following relation and useful definition 
	\begin{align}\label{eq:not}
		\nabla_Vf= V^A\nabla_Af=V^B\partial_B\bar{f} =:V^B\frac{\partial f }{\partial X^B} =V^B\partial_B{f} .
	\end{align}

	To consider terms involving second derivatives (Hessian) we introduce in the for submanifolds of $\mathcal{E}$ the concept of a connection,  see \cite[Equation 5.4]{boumaye} and \cite[Theorem 5.2]{boumaye}

	\begin{definition}\label{def:conn}(\textbf{Connection})
		If  a (Pseudo-)Riemannian manifold  $\ma$  is a submanifold of the linear vector space $\mathcal{E}$, we define the connection as 
		\begin{align}
			\nabla_U V&= 
			\text{Proj}_X( \bar{\nabla}_U\bar{V}(X) ),
		\end{align}
		where $V$ is a smooth vector field on $\ma$, $U\in\tx\ma$ is a tangent vector at the point $X$, $\bar{V}$ is any smooth extension of $V$ and $\bar{\nabla}$ is the connection on $\mathcal{E}$. 
	\end{definition}
	This connection is the (Pseudo-)Riemannian connection on $\ma$, see  \cite[Theorem 5.9]{boumaye}.   
	The following result allows us to express the Hessian more explicitly, \cite[Corollary 5.16.]{boumaye}.
	\begin{corollary}
		Let $\ma$  be a  (Pseudo-)Riemannian submanifold of a  (Pseudo-)Euclidean space.
		Consider a smooth function $f:\ma\rightarrow\R$. Let $\bar{G}$ be a smooth extension of $\nabla f$—
		that is, $\bar{G}$ is any smooth vector field defined on a neighborhood of $\ma$ in the
		embedding space such that $\bar{G}(X)=\nabla f(X)$ for all $X\in\ma$. Then,
		\begin{align}
			\text{Hess}f(X)[U]=Proj_X(\bar{\nabla}_U\bar{G}(X)  ).
		\end{align} 
	\end{corollary}
	Using the former Definition and Corollary, we can write the Hessian as 
	\begin{align}\nonumber
		\nabla_U\nabla_B f(X)&=Proj_X(U^A\partial_A \nabla_B f(X)  )\\&=
		P_B^{\,\,C}(U^A\partial_A \nabla_C f(X)  ).
	\end{align}

	In order to define the Rieffel deformation we want to move away from a spacetime point $x\in\mathcal{M}$ along a tangent vector, while remaining on the manifold $\mathcal{M}$. For that reason, we introduce the concept of retraction at the point $x\in\mathcal{M}$, see \cite[Definition 3.47.]{boumaye} and \cite[Definition 4.1.1]{mahony}
	\begin{definition}
		A retraction on a  manifold $\mathcal{M}$ is a smooth map $R$ from the tangent bundle $T\mathcal{M}$ onto the manifold $\mathcal{M}$  
		\begin{align*}
			R:T\mathcal{M}\rightarrow\mathcal{M}:(x,v)\mapsto R_x(v),
		\end{align*} 
		such that each curve $c(t) = R_x (tv)$ satisfies $c(0) = x$ and $c'(0) = v$.
	\end{definition}

	If $\mathcal{M}$ is an embedded manifold of a vector space $\mathcal{E}$, the retraction can be understood as follows. Since $T_X\mathcal{M}$ is a linear subspace of $T_X\mathcal{E}$ which is identified with  $\mathcal{E}$, we can consider the sum $X+\xi_X$, where $X\in\mathcal{M}$ viewed as a point in $\mathcal{E}$ and $\xi_X\in T_X\mathcal{E}$. The sum is well defined as an element in $\mathcal{E}$ (with a slight abuse of notation). Projecting that element back to the embedded manifold $\mathcal{M}$ gives the sum as an element in $\mathcal{M}$. This procedure defines a  retraction. If the manifold $\ma$ is a Riemannian or pseudo-Riemannian manifold it admits a (second-order) retraction that is defined by the exponential mapping, see \cite[Definition 5.42. and Proposition 10.18]{boumaye}. Yet, this is computationally challenging, and hence we give in the following an algorithm given in \cite[Section 4.1.1.]{mahony} to calculate a retraction that is more natural in the aspect of translations. This is done to keep close to the original definition of Rieffel, i.e. for the case where the manifold is Minkowski, we should get the definition of the deformed integral via translations. 
	
	\begin{definition}
		A second-order retraction $R$ on a (Pseudo-)Riemannian manifold $\ma$ is
		a retraction such that, for all $x\in\ma$ and all $v\in T_x\ma$, the curve $c(t) = R_x(tv)$ 
		has zero acceleration at $t = 0$, that is, $c''(0) = 0$.
	\end{definition}
	Using the gradients, hessians, projectors and the second-order retraction  we have a Taylor expansion on curves given by, \cite[Proposition 5.44.]{boumaye}
	\begin{align}
		f(R_x(v))=f(x)+(v,\nabla\,f(x))_x+
		\frac{1}{2}(v,\nabla_v\nabla f(x))_x+
		\mathcal{O}(v^3).
	\end{align}
	\begin{definition}
		Let $(\ma,g)$ be a globally hyperbolic spacetime, and let $T\ma$ denote the tangent bundle of the manifold. Then, we denote by $\mathcal{G}_{\ma}$ the set of transformations that are generated by second-order retractions $R:T\ma\rightarrow\ma$. 
	\end{definition}
	\begin{definition}
		The set $\mathcal{G}_{\ma}$ forms a group by the natural identification of the group decomposition denoted by the symbol $\circ$
		\begin{align}
			R_{R_X(V)}(U)  =  R_X(V\circ U),
		\end{align}
		for $X\in\ma$ and $V\in\tx\ma,U\in T_{R_X(V)}\ma$. 
		The associativity of the composition rule is a direct consequence of the associativity inherent in the composition of maps and the identity element follows   by a property of the retraction $R_X(0)=X$.
	\end{definition}
	
	\begin{remark}
		The definition of the group $\mathcal{G}_{\ma}$ is sensible since one could build retractions from the exponential map, which under certain conditions can be considered as a map $exp:{\mathfrak {g}}\rightarrow G$ from a Lie algebra ${\mathfrak {g}}$ (thought of as the tangent space to the identity element of the group $G$)  to a Lie group $G$.    See, for example, \cite{mueck} for the proof of the group properties w.r.t.\ the exponential map.
	\end{remark}


\subsection{Poisson Structure}
A Poisson manifold  is defined by $(\ma,\{\cdot,\cdot\})$, where $\ma$ is smooth manifold  and $\{\cdot,\cdot\}$ is a Poisson bracket  that can be equivalently represented by a Poisson bivector, see \cite[Chapter 1.1]{poison} or \cite[Proposition 4.1.20]{waldpoiss}.
\begin{definition}\label{def:poisbivec}
	A \textit{Poisson bivector} on a smooth manifold  $\ma$ is a smooth bivector field $\pi\in \Gamma^{\infty}(\Lambda^2(T\ma))$, i.e.\ $\pi$ is a smooth skew-symmetric tensor,   where $\Lambda^2$ denotes the exterior  power of the tangent bundle and $\Gamma^{\infty}$ denotes the space of smooth sections,  satisfying the non-linear partial differential equation \cite[Equation  1.5]{poison} (Jacobi identity) 
	\begin{align*}
		{\pi}^{AB}  \partial_A {\pi}^{CD}+  {\pi}^{AC}  \partial_A  {\pi}^{DB}+  {\pi}^{AD}  \partial_A  {\pi}^{BC}=0.
	\end{align*} 
	Choosing local coordinates $(U,X)$, any Poisson bivector is given by $$
	{\displaystyle \pi _{\mid U}=\frac{1}{2}\sum _{A,B}\pi ^{AB}{\frac {\partial }{\partial X^{A}}}\wedge{\frac {\partial }{\partial X^{B}}}.}$$ The connection between the Poisson bracket and the Poisson bivector is provided by 
	\begin{equation*}
		\{f,g\}_{\pi}= \sum _{A,B}\pi ^{AB} {\partial_A f } \,{\partial_B g}.
	\end{equation*}
\end{definition}  \begin{remark}
	With a slight abuse of notation   we drop the bar symbol denoting the smooth extensions, where there is no risk of confusion.   
\end{remark} 
In the following result we make a connection between the product of a Poisson bivector and the orthogonal projector. 
\begin{lemma}\label{lem:thetproj}
	For a smooth bivector field $\pi\in \Gamma^{\infty}(\Lambda^2(T_Z\ma))$ and the orthogonal projector the following relations hold, 
	\begin{equation*}
		\pi^{AB}(Z)P_{B}^{\,\,C}(Z)= \pi^{AC}(Z),\qquad\qquad        \pi^{BD}(Z)P_{B}^{\,\,C}(Z)= \pi^{CD}(Z) ,
	\end{equation*} 
	\begin{equation*}
		\pi^{AD}(Z) \pi^{BC}(Z) \partial_CP_{D}^{\,\,E}(Z)=    \pi^{AC}(Z) \pi^{BD}(Z) \partial_CP_{D}^{\,\,E}(Z).
	\end{equation*}
\end{lemma}
\begin{proof}
	We use the orthogonality condition given in Definition \ref{defproj} and write an element $V\in T_Z\ma$ as 
	\begin{align*}
		V^A=\pi^{AC}W_C.
	\end{align*}
	Then, the orthogonality condition for any element $U\in\R^N$ reads 
	\begin{align*}
		(U_A-P_{A}^{\,\,B}U_B)\pi^{AC}W_C=0,
	\end{align*}
	which reads
	\begin{align*}
		( \pi^{AC}  -P_{B}^{\,\,A}\pi^{BC})U_AW_C=0,
	\end{align*}
	since this holds for all $U\in\R^N$ and $W\in T_Z\ma$, the proof of the first statement  is concluded. By the skew-symmetry of the Poisson bivector and the first relation,  the second statement is proved, 
	\begin{align*}
		\pi^{BC}P_{B}^{\,\,A}=-\pi^{CB}P_{B}^{\,\,A}=-\pi^{CA}=\pi^{AC}.
	\end{align*}
	For the second part, we use the first part of the lemma and the Jacobi identity, i.e.\ we write the equation as
	\begin{align*}
		\pi^{AD}  \pi^{BC} \partial_CP_{D}^{\,\,E} &=    \pi^{AC}  \pi^{BD} \partial_CP_{D}^{\,\,E} 
		\\  \pi^{BC}\partial_C (  \pi^{AD}P_{D}^{\,\,E})- \pi^{BC}P_{D}^{\,\,E} \partial_C   \pi^{AD}&=    \pi^{AC}   \partial_C(\pi^{BD}P_{D}^{\,\,E} )- \pi^{AC} P_{D}^{\,\,E} \partial_C\pi^{BD} \\
		\pi^{BC}\partial_C   \pi^{AE}  - \pi^{BC}P_{D}^{\,\,E} \partial_C   \pi^{AD}&=    \pi^{AC}   \partial_C \pi^{BE}  - \pi^{AC} P_{D}^{\,\,E} \partial_C\pi^{BD},
	\end{align*}
	Rearranging the terms we have, 
	\begin{align*} 
		\pi^{BC}\partial_C   \pi^{AE}+\pi^{AC}   \partial_C \pi^{EB}   &= P_{D}^{\,\,E} ( \pi^{BC}\partial_C   \pi^{AD}    + \pi^{AC}   \partial_C\pi^{DB}).
	\end{align*}   
	Next, we use the Jacobi identity to which the equation reduces to
	\begin{align*}
		-  \pi^{EC}\partial_C   \pi^{BA}= -P_{D}^{\,\,E} \pi^{DC}   \partial_C\pi^{BA}=-\pi^{EC}   \partial_C\pi^{BA}.
	\end{align*}Reading the equalities from the bottom to the top we conclude the proof.
\end{proof} 
\begin{lemma}\label{lemmpi0}
	The covariant derivative of the Poisson bivector field $\pi\in \Gamma^{\infty}(\Lambda^2(T_Z\ma))$ is given by 
	\begin{align*}
		U^A\nabla_A\pi^{DC}=U^A\partial_A\pi^{DC}-U^A\pi^{EC}\partial_A  P^{\,\,D}_E-U^A\pi^{DE}\partial_A  P^{\,\,C}_E.
	\end{align*}
\end{lemma}
\begin{proof}
	
	Note that with the use of the orthogonal projection w.r.t.\ the embedding point $Z$, we can write the covariant derivative of a vector field $V\in   T_Z\ma$, using Definition \ref{def:conn}, as 
	\begin{align*}
		U^A\nabla_A V^B&=
		U^AP^{\,\,B}_C\partial_A \bar V^C\\
		&=  U^A\partial_A (P^{\,\,B}_C \bar V^C)- U^AV^C \partial_A  P^{\,\,B}_C \\
		&=  U^A\partial_A   \bar V^B- U^AV^C \partial_A  P^{\,\,B}_C .
	\end{align*}
	Applying this using standard methods (see for example \cite[Section 4.7]{boumaye} to the two-tensor $\pi$  gives us 
	\begin{align*}
		U^A\nabla_A\pi^{DC}=U^A\partial_A\bar \pi^{DC}-U^A\pi^{EC}\partial_A  P^{\,\,D}_E-U^A\pi^{DE}\partial_A  P^{\,\,C}_E,
	\end{align*}
	and using the notation in Equation \eqref{eq:not} we conclude the proof.
\end{proof}

\section{A Generalized Rieffel Deformation}
Using the embedding formalism,   we define the Rieffel product for  a   globally hyperbolic spacetime $(\mathcal{M},g)$ with a certain Poisson structure $\pi$,  as a smooth action of the  group $\gm$. See~\cite{R} for the case $\mathcal{G}_{\R^4}=\R^4$ and \cite{bahnwald, wald1, DLS} for    constructions that use the exponential map to define a deformation quantization. The smooth action $\alpha$ of $\gm$ acts as a retraction  in the embedding coordinates on functions extended to the embedding space. 

\subsection{Oscillatory Integrals}
In order to define the Rieffel product a few   definitions and results w.r.t.\ oscillatory integrals are  briefly given. First, we define an oscillatory integral as in Ref.~\cite[Ch.~7.8]{H}.
\begin{definition}
	\label{oipf}
	Let $X \subset \mathbb{R}^M$ be open and let $\Gamma$ be an open cone on $X \times \left( \mathbb{R}^N \setminus \{0\} \right)$ for some $N$, i.e., for all $(x,y) \in \Gamma$ also $(x,\lambda y) \in \Gamma$ for all $\lambda > 0$. Then an integral of the form (see Ref.~\cite[Eq.~7.8.1]{H})
	\begin{equation}
		\label{eq:iopf_integral}
		\int \mathe^{-\mathi    x \cdot y} \,b(x,y) \total^N y,
	\end{equation}
	is called an \textbf{oscillatory integral}.
\end{definition}
Next, we define a symbol, see Ref.~\cite[Def.~7.8.1]{H}.
\begin{definition}
	Let $X\subset \R^M$ be open and let $m$, $\rho$, $\delta$ be real numbers with $0 < \rho \leq 1$ and $0 \leq \delta < 1$. Then we denote by $S^{m}_{\rho,\delta}(X\times \mathbb{R}^N)$ the set of all $b\in C^\infty\left( X \times \mathbb{R}^N \right)$ such that for every compact set $K \subset X$ and all multi-idnices $\alpha$, $\beta$ there exists some constant $C_{\alpha,\beta,K}$ such that the estimate
	\begin{equation}
		\abs{ \partial^{\beta}_x \partial^{\alpha}_y b(x,y) } \leq C_{\alpha,\beta,K}(1+\abs{y})^{m - \rho \abs{\alpha} + \delta \abs{\beta}},
	\end{equation}
	is valid for all $(x,y) \in K \times \mathbb{R}^N$. The elements $S^{m}_{\rho,\delta}$ are called symbols of order $m$ and type $\rho,\delta$.
\end{definition} 
\begin{remark}\label{remwd}
	It is proven in Refs.~\cite[Thm.~7.8.2]{H}, (see also Refs.~\cite{LW}), that if $b \in S^{m}_{\rho,\delta}$ and $m < -N$ the oscillatory integral~\eqref{eq:iopf_integral} converges absolutely. In the case $m \geq -N$, the oscillatory integral has to be defined in a distributional manner. Using \cite[Proposition 1.1.11, Lemma 1.2.1, Proposition 1.2.2 ]{HOSC} one can prove the convergence of the oscillatory integral whenever $b(x,y)$ is a symbol even for the case $m \geq -N $ by using a mollifiying function. 	 
\end{remark} 
\begin{definition}
	We define the domain $\mathcal{D}^{m}_{\rho,\delta}(\mathcal{M})$ 
	\begin{align}
		\mathcal{D}^{m}_{\rho,\delta}(\mathcal{M}):=\{f,g\in C^{\infty}(\mathcal{M})|\,\alpha_{\Theta X}(f) \alpha_Y(g) \in S^{m}_{\rho,\delta}, \text{with } X,Y\in T\mathcal{M}\}.
	\end{align} 
	as the collection of a  pair of  two smooth functions $(f,g)$   that are symbols w.r.t.\ the  retractions.
	Analogously, we define the   domain $\mathcal{D}^{m}_{0,\rho,\delta}(\mathcal{M})$ 
	\begin{align}
		\mathcal{D}^{m}_{0,\rho,\delta}(\mathcal{M}):=\{f,g\in C^{\infty}_0(\mathcal{M})|\,\alpha_{\Theta X}(f) \alpha_Y(g) \in S^{m}_{\rho,\delta}, \text{with } X,Y\in T\mathcal{M}\}.
	\end{align} 
	as the collection of a  pair of two smooth functions with compact support $(f,g)$   that are symbols w.r.t.\ the  retractions.
\end{definition}
\subsection{The Deformed Product}
\begin{definition}
	\label{def:defpro}
	Let the smooth action $\alpha$ of the group $\gm$ denote second-order retractions w.r.t\ the manifold $(\mathcal{M},g)\subset (\mathbb{R}^N,\eta )$ (or $(\R^N,\delta )$) and let $\theta$ be a real, constant deformation parameter and $\pi\in \Gamma^{\infty}(\Lambda^2(T_Z\ma))$ a Poisson bivector that define the matrix $\Theta := \theta \,  \pi   $.  
	Then, the (\textit{formal}) generalized Rieffel product of two functions $(f,g)\in  \mathcal{D}^{m}_{\rho,\delta}(\mathcal{M})$ is defined as
	\begin{align*}
		\label{eq:defpro}
		\left( f \star_\theta g \right)(z) &\equiv  \lim_{\epsilon \to 0} \iint {\chi}(\epsilon X, \epsilon Y) \, \alpha_{\Theta   X}(f(Z))  \, \alpha_{Y}(g(Z))   \,\mathe^{- {\mathi} \, (  X,    Y )_z   } \total^N X \total^N Y \\
		&=  \lim_{\epsilon \to 0} \iint \chi(\epsilon X, \epsilon Y) \, f(R_{Z}( \Theta X))  \, g(R_{Z}( Y)) \,  \,\mathe^{- {\mathi} \,X \cdot\, Y  }  \total^N X \total^N Y\eqend{,}
	\end{align*}
	where $Z$ is the embedding point corresponding to   $z$,       $X,Y\in T_Z\ma$ and the integrations are w.r.t.\ the non-vanishing components (i.e.\ maximally $N$) and  the scalar product $  \cdot  $ is w.r.t.\ the flat Minkowski (or Euclidean) metric $\eta_{AB}$ ($\delta_{AB}$) at the point $Z$. Moreover,  the cut-off function $\chi \in   C_0^{\infty}(\R^N   \times  \R^N)$ is chosen such that   condition $\chi(0,0) = 1$ is fulfilled. 
\end{definition}

By identifying  $\alpha_{\Theta X}(f(Z)) \, \alpha_{Y}(g(Z))$ with the function $b(X,Y)$ in Definition \ref{oipf}, it is clear that the Rieffel product is given as an oscillatory integral.

The generalized Rieffel product satisfies the following   properties, where the first two are standard for a star product, see \cite[Defintion 6.1.1]{waldpoiss} and \cite{Kont}. 
\begin{proposition}
	For   functions $(f, g) \in   \mathcal{D}^{m}_{\rho,\delta}(\mathcal{M})$ the generalized Rieffel product given in Definition~\ref{def:defpro} is well-defined and satisfies the following properties \newline
	\begin{itemize}
		\item Unital, \ $$1\star_\theta f=f\star_\theta 1=f$$ ,\newline
		\item 
		The commutative limit,  $$\displaystyle\lim_{\theta \rightarrow 0} (f \star_\theta g) (z)= (f \cdot g) (z),$$  	\newline
		\item The flat limit, i.e.\ in case that the manifold is the four-dimensional  flat Minkowski (or Euclidean) space, with retractions provided by the translation group and  constant  Poisson bivector the generalized Rieffel product turns to the standard Rieffel product.
	\end{itemize}   
	\end{proposition} 
	\begin{proof}
		To see that the Rieffel product is well-defined, for functions $(f, g) \in   \mathcal{D}^{m}_{\rho,\delta}(\mathcal{M})$,  we refer the reader to  results mentioned in Remark \ref{remwd}.  Unitality is seen by a straight-forward calculation, 
		\begin{align*}
			\left( f \star_\theta 1 \right)(z) &\equiv  \lim_{\epsilon \to 0} \iint {\chi}(\epsilon X, \epsilon Y) \, \alpha_{\Theta   X}(f(Z))  \, \alpha_{Y}(1)   \,\mathe^{- {\mathi} \, (  X,    Y )_z   } \total^N X \total^N Y \\&=     \lim_{\epsilon \to 0} \iint {\chi}_1(\epsilon X) \, f(R_{Z}( \Theta X))  \, \delta^N (  X ) \total^N X  \\&=  f(Z),
		\end{align*}
		where we used  the property $R_Z(0)=Z$. The proof for $\left(1\star_\theta f \right)(z)$ is analogous.

		The    commutative limit $\lim_{\theta \rightarrow 0}$ is proven as follows. We distinguish between two cases. The first case where the (translated) functions $\alpha_{\Theta X}(f) \alpha_Y(g)$ are symbols with $m<-4$ is trivial since for such symbols no mollifier   is needed. In this case we interchange the limit $\lim_{\theta \rightarrow 0}$ with the integral by using  the dominated Lebesgue convergence theorem, since we know that the function inside the integral is bounded (by the very definition of a symbol) by a measurable function. The explicit term is then given by, 
		\begin{align*}
			\lim_{\theta \rightarrow 0}	\left( f \star_\theta g \right)(z)  &=
			\lim_{\theta \rightarrow 0} \iint  \, \alpha_{\Theta X}(f(Z)) \, \alpha_{Y}(g(Z)) \, \,\mathe^{- {\mathi}  X  \cdot\,Y   }  \total^N X \total^N Y
			\\&=f(Z)  \,   \iint  \alpha_{Y}(g(Z))\, \delta (  Y  ) \total^N Y \\&=f(z)\cdot  g(z),
		\end{align*}
		where we used the continuity of $f$ and the property $R_Z(0)=Z$ of the retraction. In the case where the function $\alpha_{\Theta X}(f) \alpha_Y(g)$ is a  symbol with $m\geq-4$ the deformed product is defined  in a distributional manner, which we denote by  
		$$T(\varphi_{\theta}):=\lim_{\epsilon \to 0}\iint \,e^{-iX\cdot Y} \,\varphi_{\theta}^{\epsilon}(X,Y)\,d^NX\,d^NY.$$
		By continuity of the distribution, it is clear that $T(\varphi_{\theta})$ converges to $T(\varphi_{0})$ in the limit $\theta\rightarrow0$, since $\varphi_{\theta}\rightarrow\varphi$ in $ C_0^{\infty}(\R^4   \times  \R^4)$.  The commutative limit is as well given in  \cite[Corollary 2.8]{R} for functions belonging to certain Fr\'echet spaces.\par
		We can (trivially)  embed $(\mathcal{M},g)=(\R^4,\eta)$  into a one dimension higher flat space $(\R^5,\eta)$ by the embedding function $F=Z^5$, where the tangent space is   $T_Z \R^{4} = \R^{4}$. Since the $5$-th component vanishes for both tangent vectors, we only integrate over four components. Taking the retractions to be translations we have the deformed integral  
		\begin{align*}
			\left( f \star_\theta g \right)(z) &\equiv  \lim_{\epsilon \to 0} \iint {\chi}(\epsilon X, \epsilon Y) \, \alpha_{\Theta  X}(f(Z))  \, \alpha_{Y}(g(Z))   \,\mathe^{- {\mathi} \,  X \cdot\,Y  } \total^4 X \total^4 Y\\&=\lim_{\epsilon \to 0} \iint {\chi}(\epsilon x, \epsilon y) \, f(z+\Theta  x)  \, g(z+y)   \,\mathe^{- {\mathi}   x \cdot\,y  } \total^4 x \total^4 y .
		\end{align*} 
	\end{proof} 
			\label{2.1}
			
		Further properties that are standard for a star product, namely associativity and Poisson-compatibility, are satisfied and are subject of the following proposition.
		
		\begin{proposition} \label{prop:associative} 
			The generalized  Rieffel product  is Poisson compatible,  
			\begin{align*}
				\frac{i}{2}( f\star_\theta g-  g\star_\theta f ) &=  \theta \,  \{f,g\}_{\pi},
			\end{align*} 
			and hence associative up to first order in $\theta$, i.e.,
			\begin{equation}
				\left( \left( f \star_\theta g \right) \star_\theta h \right)(z) = \left( f \star_\theta	\left( g \star_\theta h \right) \right)(z)+\mathcal{O}(\theta^2) .
			\end{equation}
			
		\end{proposition} 
		\begin{proof}
			We prove the proposition by using the Taylor expansion of the retractions  in order of the deformation parameter $\theta$, i.e.\ 
			\begin{align*}
				f(R_{Z}( \Theta X))=f(Z)+(\Theta X)^A\,\nabla_Af(Z)+\mathcal{O}(\theta^2),
			\end{align*}
			where we use the following notation
			\begin{align}
				\nabla_A f(Z_1)=  \pz (\partial_A \bar{f}(Z_1)) .
			\end{align}  
			Plugging this into the deformation formula we have up to first order in $\theta$, 
			\begin{align*}
				\left( f \star_\theta g \right)(z) &=
				\lim_{\epsilon \to 0} \iint \chi(\epsilon X, \epsilon Y) \, f(R_{Z}( \Theta X))  \, g(R_{Z}( Y)) \,  \,\mathe^{- {\mathi} \,X \cdot\, Y  }  \total^N X \total^N Y\eqend{,}\\ &=
				\lim_{\epsilon \to 0} \iint \chi(\epsilon X, \epsilon Y) \, (f(Z)+(\Theta X)^A\,\nabla_A f (Z)) \, g(R_{Z}( Y)) \,  \,\mathe^{- {\mathi} \,X \cdot\, Y  }  \total^N X \total^N Y\eqend{,}\\ &=f (z)  g(z) + \Theta^{AB}	 \,\nabla_A f(Z)  \lim_{\epsilon \to 0} \iint \chi(\epsilon X, \epsilon Y) \,X_B    \, g(R_{Z}( Y)) \,  \,\mathe^{- {\mathi} \,X \cdot\, Y  }  \total^N X \total^N Y\eqend{,}\\  &=f (z)  g(z) -i \Theta^{AB}	 \,\nabla_A f   \,\nabla_B g    ,
			\end{align*}
			where in the last lines we used the form of the retractions $g(R_{Z}( Y))=g(Z)+Y^A\,\nabla_A g(Z)+\mathcal{O}(Y^2)$. To prove associativity we first consider 
			\begin{align*}
				&   ((f\star_{\theta}g) \star_{\theta}h )(z) = (F \star_{\theta}h)(z)  =
				F (z)  h(z) -i \Theta^{AB}	 \,\nabla_A F   \,\nabla_B h   
				\\&\,\\&=(f (z)  g(z) -i \Theta^{AB}	 \,\nabla_A f   \,\nabla_B g     )h(z)  -i \Theta^{AB}	 \, (\nabla_A f (z)  g(z)+f (z)  \nabla_A g(z) )   \,\nabla_B h  .
			\end{align*} 
			Next we consider, 
			\begin{align*} &  (f\star_{\theta}(g \star_{\theta}h ))(z)  = (f\star_{\theta}H)(z)   = f (z) H(z) -i \Theta^{AB}	 \,\nabla_A f   \,\nabla_B H        \\&\,\\&=f (z) (g (z)  h(z) -i \Theta^{AB}	 \,\nabla_A g  \,\nabla_B h)   
				-i \Theta^{AB}	 \,\nabla_A f   \, (\nabla_B g (z)  h(z)+g (z)  \nabla_B h(z)   )  .
			\end{align*} Using this expression we have for the commutator 
			\begin{align*}
				\frac{i}{2}( f\star_\theta g-  g\star_\theta f ) &= \Theta^{AB}	 \,\nabla_A f   \,\nabla_B g \\&= \theta\,\pi^{AB}	 \,\nabla_A f   \,\nabla_B g \\&= \theta\,\pi^{AB}	 \,\partial_A f   \,\partial_B g,
			\end{align*}
			where in the last line we used Lemma \ref{lem:thetproj}.
		\end{proof}  
		
		Next, we turn to the issue of associativity of the generalized product.  
		\begin{theorem}\label{thm:defprodorder}
			Let the Poisson bivector  $\pi\in \Gamma^{\infty}(\Lambda^2(T_Z\ma))$,  have a vanishing covariant derivative, 
			\begin{align*}
				U^A\nabla_A\pi^{CD}(Z)=0,
			\end{align*}
			for all $U\in T_Z\ma$ and $Z\in\ma$.
			Then, the generalized  Rieffel product  is associative up to second order in $\theta$, and it is explicitly given by 
			\begin{align*}
				&\left( f \star_\theta g \right)(z)   =f    g -i \Theta^{AB}	 \,\partial_A f  \,\partial_B g  - \frac{1}{2} 
				\Theta^{AC}\Theta^{BD}
				\partial_A\nabla_B f \, \partial_C\nabla_D g ,
			\end{align*} 
			for  functions $(f,g)\in  \mathcal{D}^{m}_{\rho,\delta}(\mathcal{M})$.
		\end{theorem}
		
		\begin{proof}
			See Appendix \ref{appthm:defprodorder}.
		\end{proof}
		A few remarks regarding the vanishing covariant derivative of the Poisson tensor, in order to assure associativity, are in order. A triple $(\mathcal{M},\pi,g)$ with the additional condition $\nabla \pi=0$, is called a (Pseudo-)Riemannian Poisson manifold, see \cite[Defintion 1.1.]{boum} and references therein. For a  definition of Poisson-Riemannian manifold, where the covariant constant  condition on the Poisson tensor is called Poisson compatibility, see \cite[Equation 3.1]{majbeg}.

		In  case, where the Poisson tensor has maximal rank $2n$ equal to the dimension of the manifold, the inverse of $\pi^{AB}$ is given by the matrix $\omega_{AB}$, a symplectic structure, which defines a symplectic manifold. In this case, Fedosov, \cite{fedo1} has proven the existence  of a deformation quantization if there is a  symplectic connection, which  is a torsion-free connection
		that preserves the symplectic tensor, defined by condition, \cite[Definition 2.3]{fedo1}.
		\begin{align*}
			\nabla_{C}\omega_{AB}=0.
		\end{align*}
		A collection of a symplectic manifold and a symplectic connection is referred to as a Fedosov manifold, see \cite{fedo2}. We proved associativity of the generalized deformed product, up to second order in the deformation parameter,  if the covariant derivative of the Poisson structure vanishes. In the symplectic case, the vanishing Poisson structure induces a symplectic connection, which can be seen from the following identity,
		\begin{align}\label{eqfedo}
			\nabla(\pi \omega )= 0,
		\end{align}
		the vanishing covariant derivative of the Poisson tensor and the Leibniz rule.

		\begin{example}
			Let for the  two-sphere $\mathbb{S}^2$  the Poisson bivector at the embedding point $Z$  and the  orthogonal projection be assigned by, 
			\begin{align}
				\pi^{AB}(Z)=\varepsilon^{ABC}Z_C,\qquad \qquad P_A^{\,\,C}(Z)=\delta_A^{\,\,C}-Z_A\,Z^{C}.
			\end{align} 
			Then, $ U^A\nabla_A\pi^{AB}=0$.  
		\end{example}
		\begin{proposition}\label{prop:ncspacetime}
			The commutation relations between the embedding coordinates are, up  to second order in the deformation parameter $\theta$,  explicitly presented by
			\begin{align*}
				[X^A,X^B]_{\theta}:=   X^A\star_\theta X^B-X^B\star_\theta X^A=-2i\,\Theta^{AB}.
			\end{align*}
		\end{proposition}
		\begin{proof}
			Using Theorem \ref{thm:defprodorder} we have for the deformed product of the coordinates $X$
			\begin{align*}
				X^A\star_\theta X^B&=      X^A    X^B -i \Theta^{CD}	 \,\partial_C X^A  \,\partial_D X^B - \frac{1}{2} 
				\Theta^{EC}\Theta^{FD} 
				\partial_E\nabla_F X^A \, \partial_C\nabla_D X^B
				\\&=      X^A    X^B -i \Theta^{CD}	 \,\delta_C^{\,\,A}\,\delta_D^{\,\,B} - \frac{1}{2} 
				\Theta^{EC}\Theta^{FD} 
				(  \partial_E P_F^{\,\,R}  \,\delta_R^{\,\,A}) \,   \partial_C  P_D^{\,\,S}  \,\delta_S^{\,\,B} \\&=      X^A    X^B -i \Theta^{AB}	   - \frac{1}{2} 
				\Theta^{EC}\Theta^{FD} \, \,
				(  \partial_E P_F^{\,\,A}  ) \,   \partial_C  P_D^{\,\,B}  .
			\end{align*}
			Next, we consider the deformed commutator 
			\begin{align*} [X^A,X^B]_{\theta} =-2i \Theta^{AB}-\frac{1}{2} \Theta^{EC}\Theta^{FD} \, \,
				\left( 
				(  \partial_E P_F^{\,\,A}  ) \,   \partial_C  P_D^{\,\,B}- (  \partial_E P_F^{\,\,B}  ) \,   \partial_C  P_D^{\,\,A}  \right),
			\end{align*}
			by interchanging the indices $E$ and  $C$ and as well $F$ and  $D$, the second order term cancels.
		\end{proof} 
		
		\begin{remark}
			Following the definition of the generalized deformed product and taking into account the the well-know
			fact that $\{X^A,X^B\}=\pi^{AB}$, the noncommutativity of the coordinates, given by
			\begin{align*}
				[X^A,X^B]_{\theta}=-2i\theta \pi^{AB},
			\end{align*}
			emerges in a  canonical fashion. For example in the flat case, the canonical Poisson bivector is presented by the canonical, constant, skew-symmetric tensor and the noncommutativity has the Moyal-Weyl structure,
			\begin{align*}
				[x^{\mu},x^{\nu}]_{\theta}=-2i\theta \pi^{\mu\nu},
			\end{align*}
			where the emebedding is determined by the trivial one. In the case of the two-sphere the canonical induced Poisson structure gives us the well-studied fuzzy sphere, i.e.\
			\begin{align*}
				[X^A,X^B]_{\theta}=-2i\theta \,\varepsilon^{AB}_{\,\,\,\,\,\,\,\,C}X^C.
			\end{align*}
		\end{remark}

		\subsection{Deformed Product for Two Different Points}
		In the previous section we defined a generalization of the Rieffel product to the case of non-flat manifolds and established their associativity under a certain assumption on the Poisson bivector. Next, we broaden this framework to encompass the deformed product of functions defined at two distinct points. The main motivation  for this extension is the application of those methods to quantum field theory and in particular to the two-point function. 
		\begin{definition}
			\label{def:defpro1}
			Let the smooth action $\alpha$ of the group $\gm$ denote second-order retractions w.r.t.\ the manifold $(\mathcal{M},g)\subset (\mathbb{R}^N,\eta )$ (or $(\R^N,\delta )$) and let $\theta$ be a real constant deformation parameter and $\pi\in \Gamma^{\infty}(\Lambda^2(T_{Z_1}\ma))$ a Poisson bivector that define the matrix $\Theta := \theta \,  \pi   $.  
			Then, the (\textit{formal}) generalized Rieffel product of two functions $(f,g)\in  \mathcal{D}^{m}_{\rho,\delta}(\mathcal{M})$ at two different points is defined as\begin{align}\nonumber
				\label{eq:defpro}
				f (z_1) \star_\theta g (z_2)  &\equiv  \lim_{\epsilon \to 0} \iint {\chi}(\epsilon X, \epsilon Y) \, \alpha_{\Theta \,  \delta({Z_1,Z_2})  X}(f(Z_1))  \, \alpha_{Y}(g(Z_2))  \mathe^{- {\mathi} \, (  X,    Y )_{z_2}   }   \total^N X \total^N Y \\
				&=  \lim_{\epsilon \to 0} \iint \chi(\epsilon X, \epsilon Y) \, f(R_{Z_1}( \Theta\,   \delta({Z_1,Z_2}) X))  \, g(R_{Z_2}( Y)) \,  \,\mathe^{- {\mathi} \,X \cdot\, Y  }  \total^N X \total^N Y\eqend{,}
			\end{align}
			where $Z_1$ and $Z_2$ are the embedding point corresponding to    $z_1$ and $z_2$,       $X, Y\in T_{Z_2}\ma$  and the integration is w.r.t.\ the non-vanishing components (i.e.\ maximally $N$) and  the scalar product $  \cdot  $ is w.r.t.\ the flat Minkowski (or Euclidean) metric $\eta_{AB}$ ($\delta_{AB}$). Moreover, $\delta^{\,\,\,A'}_{B}({Z_1,Z_2})$ is the operator of geodesic transport from $T_{Z_2}\ma$ to $T_{Z_1}\ma$,\footnote{We use here the same notation as in \cite{Mor03}. This bitensor is as well notated by $g^{\mu'}_{\rho}$, see for example \cite[Chapter III]{THCb}. }   the cut-off function $\chi \in   C_0^{\infty}(\R^N   \times  \R^N)$ is chosen such that condition $\chi(0,0) = 1$ is fulfilled. 
		\end{definition}
		Note that when both points are equal, i.e.\ $z_1=z_2$, the product assumes the form previously defined for a single point. The following propositions deal with the uniqueness of this definition. In particular, we prove that various sensible definitions all agree with the product given in the former Definition \ref{def:defpro1}.
		
		\begin{proposition} 
			Let $\theta$ be a real, constant deformation parameter and $\pi\in \Gamma^{\infty}(\Lambda^2(T_{Z_1}\ma))$ a Poisson bivector that define the matrix $\Theta := \theta \,  \pi$. Then, the definition of the deformed product for two different points given in Definition \ref{def:defpro1} is equivalent to the following  product,\begin{align*} 			 f (z_1) \star_\theta g (z_2)  =   \lim_{\epsilon \to 0} \iint \chi(\epsilon U, \epsilon V) \, f(R_{Z_1}( \Theta\,    U))  \, g(R_{Z_2}(\delta({Z_2,Z_1}) V)) \,  \,\mathe^{- {\mathi} \,U \cdot\, V  }  \total^N U \total^N V\eqend{,}		\end{align*}where $Z_1$ and $Z_2$ are the embedding point corresponding to    $z_1$ and $z_2$,       $U, V\in T_{Z_1}\ma$  and the integration is w.r.t.\ the non-vanishing components, where $\delta^{\,\,\,A'}_{B}({Z_2,Z_1})$ is the operator of the geodesic transport from $T_{Z_1}\ma$ to $T_{Z_2}\ma$.
		\end{proposition}
		\begin{proof} 
			To see the equivalence to Definition \ref{def:defpro1} we perform the following variable substitutions, 
			\begin{align*}
				\delta({Z_2,Z_1}) V=Y, \qquad  U=\delta({Z_1,Z_2})X,
			\end{align*}
			where $X,Y\in T_{Z_2}\ma$. Using the relation of the determinant of the parallel propagator, namely (\cite[Equation (7.7)]{poiss}),
			\begin{align*}
				\det  \delta^{\,\,\,A'}_{B}= \frac{\det(-\eta)}{\det(-\eta')}=1,
			\end{align*}
			and using the inverses of the parallel-transports (see \cite[Equation 5.7]{poiss}) the scalar product 
			\begin{align*}
				( U,    V )_{z_1}& =  \delta^{\,\,\,B'}_{A}(Z_1,Z_2)
				X_{B'}\delta^{\,\,\,A}_{C'}(Z_1,Z_2)
				Y^{C'}    \\&=
				\delta^{A}_{\,\,\,C'}(Z_2,Z_1)  \delta^{\,\,\,B'}_{A}(Z_1,Z_2)
				X_{B'}
				Y^{C'}\\&=
				(  X,    Y )_{z_2} ,
			\end{align*}
			where in the last lines we used relation $\delta^{\,\,\,A}_{C'}(Z_1,Z_2)=\delta^{A}_{\,\,\,C'}(Z_2,Z_1)$ (see \cite[Equation (5.8)]{poiss}) and the inverse relation. 
			
		\end{proof}  
		\begin{proposition} 
			Let $\theta$ be a real, constant deformation parameter and $\pi\in \Gamma^{\infty}(\Lambda^2(T_{Z_2}\ma))$ a Poisson bivector that define the matrix $\Theta^{A'B'} (Z_2) = \theta \,  \pi^{A'B'}=\delta^{\,\,\,A'}_{A}(Z_2,Z_1)\delta^{\,\,\,B'}_{B}(Z_2,Z_1)\Theta^{AB} (Z_1)$. Then, the definition of the deformed product for two different points given in Definition \ref{def:defpro1} is equivalent to the following  product,\begin{align*} 			 f (z_1) \star_\theta g (z_2)  &=   \lim_{\epsilon \to 0} \iint \chi(\epsilon X, \epsilon Y) \, f(R_{Z_1}(\delta({Z_1,Z_2}) ( \Theta(Z_2)\,   X))  \, g(R_{Z_2}(  Y))  
				\,\mathe^{- {\mathi} \,X \cdot\, Y  }  \total^N X \total^N Y\eqend{,}	\\&=
				\lim_{\epsilon \to 0} \iint \chi(\epsilon X, \epsilon Y) \, f(R_{Z_1}(\delta({Z_1,Z_2}) \Theta(Z_2)\,  \delta({Z_2,Z_1}) U))  \, g(R_{Z_2}( \delta({Z_2,Z_1})  V)) \,\\&  \qquad\qquad\qquad\qquad\qquad\qquad\qquad\qquad\qquad\qquad\qquad\qquad\times
				\,\mathe^{- {\mathi} \,U \cdot\, V  }  \total^N U \total^N V\eqend{,}	\end{align*}where $Z_1$ and $Z_2$ are the embedding point corresponding to    $z_1$ and $z_2$,       $X, Y\in T_{Z_2}\ma$, $U, V\in T_{Z_1}\ma$, and the integration is w.r.t.\ the non-vanishing components.
		\end{proposition}
		\begin{proof}
			The first equality is easily proven, since 
			\begin{align*}
				\delta({Z_1,Z_2})\delta({Z_1,Z_2}) \Theta(Z_2)=\Theta(Z_1),
			\end{align*}
			and the second equality is proven analogously to the former Proposition. 
		\end{proof}

		\begin{lemma}\label{lem:thetproj'}
			For a smooth bivector field $\pi\in \Gamma^{\infty}(\Lambda^2(T_{Z_1}\ma))$ and  the operator of the geodesic transport  $\delta({Z_1,Z_2})$   from $T_{Z_2}\ma$ to $T_{Z_1}\ma$,  we define  the following object
			\begin{align*}
				\pi^{AC'}(Z_1,Z_2)&:=\pi^{AB}(Z_1)\,	\delta^{\,\,\,C'}_{B}({Z_1,Z_2}).
			\end{align*} 
			Then, for the tensor  $\Theta(Z_1,Z_2)$    and the orthogonal projector $P(Z_2)$ the following relation  holds, 
			\begin{equation*}
				\pi^{AA'}(Z_1,Z_2)=\pi^{AD'}(Z_1,Z_2)P^{A'}_{\,\,D'}(Z_2).
			\end{equation*} 
		\end{lemma}
		\begin{proof}
			The proof is analogous to the proof of Lemma \ref{lem:thetproj}.
		\end{proof}
		
		\begin{proposition}\label{prop:defpro2}
			The deformed product for two different points, according to Definition \ref{def:defpro1}, is given up to second order in the deformation parameter as follows,
			\begin{align*}
				f (z_1)\star_\theta g (z_2)  &=   f (z_1)  g(z_2)-i    \Theta^{AB'}(Z_1,Z_2) \,\partial_A f(Z_1)   \,\partial_{B'} g(Z_2)   
				\\  &-\Theta^{ABC'D'}    \partial_A\,\nabla_B f(Z_1)\,   \partial_{C'}\nabla_{D'}g(Z_2)    +\mathcal{O}(\Theta^3),
			\end{align*}
			where we defined 
			\begin{align*}
				\Theta^{AB'}(Z_1,Z_2)&:=\Theta^{AB}(Z_1)\,	\delta^{\,\,\,B'}_{B}({Z_1,Z_2}),
				\\ \Theta^{ABC'D'} (Z_1,Z_2)& :=        \frac{1}{4}  \left( \Theta^{AC}\,  \Theta^{BD}+ \Theta^{AD}\, \Theta^{BC}  \right)
				\delta_D^{\,\,\,D'}(Z_1,Z_2)\,\delta_C^{\,\,\,C'}(Z_1,Z_2)  
				,
			\end{align*}
			for the functions $(f,g)\in  \mathcal{D}^{m}_{\rho,\delta}(\mathcal{M})$.
		\end{proposition}
		\begin{proof}
			See Appendix \ref{approp:defpro2}.
		\end{proof}

		Although the deformed product is well-defined for a single point, its extension to two points requires an additional tool of geodesic parallel transport. In the context of quantum field theory   and associated expectation values w.r.t.\ the two-point function, the inversion is the case. Specifically, the two-point function is singular for a single point, prompting the requirement for regularization.  After regularization, one takes the coinciding point limit via the  geodesic parallel transport.
		
		\section{Quantum Field Theory in Globally Hyperbolic Spacetimes}
		\label{sec3}

		\subsection{The Quantized Scalar Field }\label{se:QF}
		In this section we give a quick overview and fix the notation  of  the free quantized scalar field $\phi$ in a globally hyperbolic spacetime $(\mathcal{M},g)$.  For a more complete discussion, see \cite{WQ}. Due to the global hyperbolicity of $(\mathcal{M},g)$ there exists unique  advanced and retarded fundamental solutions $G^{\rm adv/ret}$ (``Green's operators'') 
		for the  Klein-Gordon operator $\square_g  - m^2$ defined on smooth scalar test-functions on $\mathcal{M}$, where $\square_g$ is the Laplace-Beltrami operator w.r.t. the manifold $(\mathcal{M},g)$.   Next, we define a $*$-algebra $\mathscr{A} = \mathscr{A}(\mathcal{M},g)$ generated by the elementary algebraic objects, namely the field operators $\phi(F)$, $F \in C_0^\infty (\mathcal{M}) $, and a unit element ${\bf 1}$. It is called the CCR algebra of the quantum field  $\phi(F)$ over the manifold $\mathcal{M}$ and it  fulfills (see \cite[Definition 8]{MK})
		\begin{align*}
			(i)&  \ \ F \mapsto \phi(F) \ \ \text{is}\ \mathbb{R}\text{-linear} \quad \quad (ii) \ \ \phi(( \nabla^a\nabla_a - m^2)F)= 0 \\
			(iii)& \ \ \phi(F)^* = \phi(F)  \quad \quad \quad \quad \quad (iv) \ \ [\phi(F),\phi(H)] = iE(F,H) \cdot {\bf 1}\,, \quad \quad F,H \in C_0^\infty(\mathcal{M}),\, 
		\end{align*}
		where $E(F,F')$ is the causal Green's function (or causal propagator)  defined by the following expression
		$$ E(F,H) = \int_{M} \left( F(x) (G^{\rm adv}H)(x) - F(x)(G^{\rm ret}H)(x) \right)\, d{\rm vol}_g(x) \,, \quad F,H \in C_0^\infty(\mathcal{M})\,,$$
		where $d{\rm vol}_g$ is the volume form w.r.t.\ the metric $g$.
		\par 
		An important consequence of the global hyperbolicity of the manifold $(\mathcal{M},g)$ is given by the causality condition  (see \cite[Proposition 4]{MK}).
		\begin{proposition}
			Referring to $\mathscr{A}(\mathcal{M},g)$, $\phi(F)$ and $\phi(H)$ commute if the supports of $F$ and $H$ are causally separated, i.e.\ $$ E(F,H) =0,$$  if the support of $F$ does not intersect $J^+_M(supp H)\cup J^-_M(supp H)$.
		\end{proposition}
		\subsection{States and GNS Theorem}\label{sexgns}
		Up to this point, the field operators and their generating algebra are abstract algebraic objects. To represent those algebraic objects on a  Hilbert space we use the GNS  (Gelfand-Naimark-Segal) construction. This theorem supplies us with such a representation in the following manner, see for example \cite[Chapter 1, Theorem 1]{MK} and references therein. 
		For a given state $\omega$  over the (unital) $*$-algebra $\mathscr{A}$, i.e.\ a positive $\mathbb{C}$-linear map $\omega:\mathscr{A}\rightarrow\mathbb{C}$ which is \textit{positive} ($\omega(a^*a)\geq 0$, for all $a\in\mathscr{A}$) and \textit{normalized}, one obtains a
		quadruple $(\mathcal{H}_{\omega},D_{\omega},\pi_{\omega},\Psi_{\omega})$. This quadruple consists of a complex Hilbert space $\mathcal{H}_{\omega}$, a dense subspace $\mathcal{D}_{\omega}\subset \mathcal{H}_{\omega}$, a  $*$-representation  $\pi_{\omega}:\mathscr{A}(\mathcal{M},g)\rightarrow\mathscr{L}(\mathcal{D}_{\omega})$  of $\mathscr{A}$ on $\mathcal{H}_{\omega}$ with domain $\mathcal{D}_{\omega}$ and a cyclic  and separating vector $\Psi_{\omega}$.   The field operators are   given by densely defined symmetric operators, 
		\begin{equation*}
			\phi_{\omega}(F)=\pi_{\omega}(\phi(f)):D_{\omega}\rightarrow \mathcal{H}_{\omega}.
		\end{equation*} 
		The expectation value of $n$ functions or the $n$-point function of elements of the algebra $\mathscr{A}$ is determined by
		\begin{align*}
			\omega_n( F_1, \cdots, F_n )=\langle \Psi_{\omega}|\,
			\pi_{\omega}(\phi(F_1))\cdots\pi_{\omega}(\phi(F_n)) \,\Psi_{\omega} \rangle.
		\end{align*} 
		Next, we list some basic properties of the two point function in the following proposition (\cite[Proposition 7]{MK}).
		\begin{proposition}
			Consider a state $\omega:\mathscr{A}(\mathcal{M})\rightarrow\C$ and define $P:=\square_{M}+m^2+\xi R$ (for arbitrarily fixed values of $m^2$, $ \xi\in\R$). The two-point function, $\omega_2$ satisfies the following relations  
			\begin{align}
				& \omega_2(F,F) \ge 0\,, \quad \omega_2(F',F) = \bar{\omega_2(F,F')}\,, \quad {\rm Im}\,\omega_2(F,F') = {\frac{1}{2}}E(F,F')\,, \\
				& \omega_2(PF,F') = 0 = \omega_2(F,PF')\,, \quad    \omega_2(F,F')- \omega_2(F',F)=iE(F,F'),
			\end{align}  
			for $F,F' \in C_0^\infty(\mathcal{M})$.
		\end{proposition} 
		Furthermore, we define the Hadamard parametrix, where for more details see \cite{K88},\cite[Section 3.2]{MK}.  Let $\sigma(x,x')$ be the half of the squared geodesic distance (the Synge function) between $x$ and $x'$ and let in a convex neighborhood $C$ of a  four dimensional spacetime the   parametrix   
		$H_\epsilon$ (that is a limit of the Hadamard parametrix) be given as,
		\begin{equation} H_{\epsilon}(x,y)=\frac{u(x,y)}{(2\pi)^2\sigma_\epsilon(x,y)}+ v(x,y)\log\left(\frac{\sigma_\epsilon(x,y)}{\lambda^2}\right),\label{Z}\end{equation}
		where $u,v$ are $C^\infty$ functions on $C\times C$, with $u(x,x)=1$ and 
		where
		$x,y \in C$,  $T$ is any local time coordinate increasing towards the future,
		$\lambda>0$ a  length scale
		and 
		\begin{eqnarray} \sigma_\epsilon(x,y)   \stackrel {\mbox{\scriptsize  def}} {=} \sigma(x,y) +2i\epsilon (T(x)-T(y)) + \epsilon^2\label{sigma}\:,\end{eqnarray}
		finally, the cut in the complex domain of the $\log$ function is assumed along the negative axis. 
		
		\begin{definition}
			\label{def_HadamardFormScalar}
			A (not necessarily quasifree) state $\omega$ on
			$\mathscr{A}(\mathcal{M})$ and its two point function  $\omega_2$ are {\bf  Hadamard} if $\omega_2 \in {\mathcal{D}}'(M \times M)$ 
			and  every point of $\mathcal{M}$ admits an open normal  neighborhood $C$ where 
			\begin{equation}\omega_2(x,y) - H_{0^+}(x,y)=  w(x,y)  \quad \mbox{for some }\:\:  w\in C^\infty(C\times C)\:. \label{hadamard}\end{equation}
			Here $0^+$ indicates  the standard weak distributional limit as $\epsilon \to 0^+$.
		\end{definition}  
		
		\subsection{The Microlocal Spectrum Condition and Hadamard}
		
		This section is a recollection of definitions and theorems about microlocal analysis, see \cite[Chapter 6]{CF} (and references therein) for more details. The principle idea of microlocal analysis is to use the decay properties of the Fourier-transformation, to gain information about the underlying singular structure.  This has to do with the fundamental connection between smooth functions and their rapidly decaying Fourier transformations. In the context of microlocal analysis, the wavefront set plays a fundamental role and we define it next, \cite[Defintion 6.1]{CF}.
		\begin{definition}
			A) If $u\in\mathcal{D'}(\mathbb{R}^n)$, a pair $(x,k)\in\R^n\times(\R^n\backslash\{0\})$ is a regular direction for $u$ if there exist \newline
			\begin{enumerate}
				\item $\phi\in C_0^{\infty}(\R^n)$ with $\phi(x)\neq0$ \\
				\item a conic neighborhood $V$ of $k$\\
				\item constants $C_N$, $N\in\N$
				so that 
				\begin{align*}
					\bigg| \hat{\phi u}(k) \bigg|<\frac{C_N}{1+|k|^N},\qquad \forall k\in V 
				\end{align*}
				i.e.\ $\hat{\phi u}$ decays rapidly as $k\rightarrow\infty$ in $V$.\\
			\end{enumerate}
			B) The \textit{wavefront} set of $u$ is defined to be \newline
			$$WF(u)=\{(x,k)\in\R^n\times(\R^n\backslash\{0\}): (x,k)  \text{\,is not \textit{a regular direction for}  u} \}.$$
		\end{definition}
		
		Some useful properties of the wavefront set are given in the following. \newline
		\begin{enumerate}
			\item \label{propwf1} If $f\in C^{\infty}$ it has an empty wavefront set \ $WF(f)=\emptyset.$ \newline
			\item  \label{propwf2} $WF(\alpha u+\beta v)\subset WF(u)\cup WF(v)$ for $u,v\in\mathcal{D'}(\R^n)$, $\alpha, \beta\in \C$. \newline
			\item \label{propwf3} If $P$ is any differential operator with smooth coefficients, then 
			\begin{align*}
				WF(Pu)\subset WF(u)  ,
			\end{align*}
			for any $u\in\mathcal{D'}(\R^n)$.
		\end{enumerate}$\,$\\
		Before turning to quantum field theory and applying the tools of microlocal analysis, we define $\mathcal{N}$ to be the bundle of nonzero null covectors on $\mathcal{M}$:
		\begin{align*}
			\mathcal{N}=\{(x,\xi)\in T^*M: \xi\, \text{a non-zero null at}\, p\}.
		\end{align*}
		The first application is the realization that the wavefront set of the Minkowski vacuum $\omega_2^{vac}$ is a subset of the product $\mathcal{N}^+\times\mathcal{N}^-$, i.e.\ 
		\begin{align*}
			WF(\omega_2^{vac})\subset \mathcal{N}^+\times\mathcal{N}^-,
		\end{align*}
		where 
		\begin{align*}
			\mathcal{N}^{\pm}=\{(p,\xi)\in\mathcal{N}:\xi\, \text{is future}(+)/\text{past}(-) \text{directed}\}.
		\end{align*}
		Since for QFT in globally hyperbolic spacetimes the singularity structure of the state has to be analog to the one given in the Minkowski case (due to the equivalence principle), we elevate this rule to a general context and give the following definition, \cite[Defintion 6.2]{CF} (see as well \cite{RAD, RAD2, IN3}),
		\begin{definition}
			A state $\omega$ obeys the \textbf{Microlocal Spectrum Condition} ($\mu SC$) if 
			\begin{align*}
				WF(\omega_2)\subset \mathcal{N}^+\times\mathcal{N}^-.
			\end{align*}
		\end{definition}
		Intuitively this condition states that the singular behaviour of the two-point function is of positive frequency in the first slot and negative in the second. This corresponds to the well-known frequency splitting that is given in the Minkowski case. We further state another important result in the context of states in QFT.
		\begin{theorem}\label{thm:cinf}
			If $\omega$ and $\omega'$ obey $\mu SC$ then
			\begin{align}
				\omega_2-\omega'_2\in C^{\infty}( \mathcal{M} \times  \mathcal{M}),
			\end{align}
			i.e.\, the $\mu SC$ determines an equivalence of class of states under equality of the two-point functions modulo $C^{\infty}$.
		\end{theorem}
		In particular, if we start with a field algebra (defined in the former section) and pick two states $\omega$ and $\omega'$ to represent the algebra (using the GNS-construction) it might happen that these two states are not unitary equivalent. Hence, they might lead to different (or nonequivalent) representations of the  field algebra. Yet, the significant physical content of a QFT can is the two-point function. Hence, if the two different two point functions of the different states $\omega$ and $\omega'$ fulfill the $\mu SC$ then their difference is smooth, which equates their corresponding wavefront set. This fact is  seen by the properties of the wavefront set, i.e.\
		\begin{align*}
			WF(\omega-\omega')=0,\qquad \qquad WF(\omega)= WF(\omega').
		\end{align*}
		In addition to the advantages of the wavefront set we state a theorem by Radzikowski (\cite{RAD, RAD2}) about the connection to the Hadamard condition.
		\begin{theorem}\label{thm:msceqhc}
			The $\mu SC$ is equivalent to the \textbf{Hadamard condition}.
		\end{theorem}	
		Hence, given a state that is Hadamard, which is  proven to exist for all massive Klein-Gordon theories,  the theorem tells us that all such states fulfill the physically reasonable microlocal spectrum condition.

		\section{Deformation of QFT in Globally Hyperbolic Spacetimes}
		\subsection{Representation of the Deformation}
		As before we start with a globally hyperbolic spacetime $(\mathcal{M},g)$, a $*$-algebra $\mathscr{A} = \mathscr{A}(\mathcal{M},g)$ generated by the field operators that satisfy the Klein Gordon equation and the choice of a state $\omega:\mathscr{A}(\mathcal{M},g)\rightarrow\C$, i.e. a normalized positive linear functional.  This state then determines as described in Subsection \ref{sexgns} the   quadruple $(\mathcal{H}_{\omega},D_{\omega},\pi_{\omega},\Psi_{\omega})$. The deformation, w.r.t.\ the noncommutative Rieffel product, enters the framework not via the states or the representation but via changing the point-wise product of the representations of the algebra. Summarizing, the \textit{deformed} $2$-point function is given in the following definition.  
		\begin{definition}\label{def:2ptfct}
			For a $*$-algebra $\mathscr{A} = \mathscr{A}(\mathcal{M},g)$ defined on a globally hyperbolic spacetime $(\mathcal{M},g)$ generated by Klein-Gordon fields, the (formal) deformed $2$-point function  is defined as the expectation value of the generalized Rieffel product (Defintion \ref{def:defpro1})  of the test functions  of the representations of the fields, i.e.\
			\begin{align*}
				\omega^{\Theta}_2( \phi(F_1)    \phi(F_2) )&:=\langle \Psi_{\omega}|\,
				\pi_{\omega}(\phi(F_1))\star_\theta  \pi_{\omega}(\phi(F_2)) \,\Psi_{\omega} \rangle\\&= \int  F_1(x_1)\star_\theta  F_2(x_2)\,
				\langle  \Psi_{\omega}|\,\pi_{\omega}(\phi(x_1))\pi_{\omega}(\phi(x_2))\,\Psi_{\omega} \rangle 
				\,d{\rm vol}_g(x_1)\,d{\rm vol}_g(x_2)   \\&=   \iint \lim_{\epsilon \to 0} \iint {\chi}(\epsilon X, \epsilon Y) \, \alpha_{\Theta \,  \delta({X_1,X_2})  X}(F_1(X_1))  \, \alpha_{Y}(F_2(X_2))  \mathe^{- {\mathi} \, (  X,    Y )_{X_2}   }   \total^N X \total^N Y  \\&\times\omega_2( \phi( X_1 )    \phi( X_2 ) )   \, \total^NX_1 \total^N X_2 \eqend{.}
			\end{align*}       
			for test-functions $F_1,F_2 \in C_0^\infty(\mathcal{M})$.
		\end{definition}
		
		\begin{theorem}
			The deformed  smeared  two-point function  given in Definition \ref{def:2ptfct}
			is       for functions $(F_1,F_2)\in \mathcal{D}^{m}_{0,\rho,\delta} (\mathcal{M})$ 
			well-defined.
		\end{theorem}
		\begin{proof}
			To prove that the oscillatory integral is well-defined  
			we write the (smeared) deformed two-point function as follows.
			\begin{equation}
				\label{eq:thm_scalarprod} 
				\omega^{\Theta}_2( \phi(F_1)    \phi(F_2) )     =  \lim_{\epsilon \to 0} \iint \chi(\epsilon X,\epsilon Y) \, b_\Theta(X,Y,F_1,F_2) \, \mathe^{-\mathi X \cdot Y} \total^N X \total^N Y,
			\end{equation}
			with $b_\Theta(X,Y,F_1,F_2)$ formally given by
			\begin{equation}
				b_\Theta(X,Y,F_1,F_2) = \iint\,\alpha_{\Theta \,  \delta({X_1,X_2})  X}(F_1(X_1))  \, \alpha_{Y}(F_2(X_2))  \, \,  
				\langle \Psi_{\omega}|\,
				\phi(X_1 ) \phi(X_2 )  \,\Psi_{\omega} \rangle\, \total^N X_1 \total^N X_2\eqend{.}
			\end{equation}
			Following Ref.~\cite[Proposition 1.1.11, Lemma 1.2.1, Proposition 1.2.2 ]{HOSC}  the oscillatory integral is well defined if we prove that the function $b_\Theta(X,Y,F_1,F_2)$ belongs to the symbol space and thus we have for   $\abs{ b_\Theta(X,Y,F_1,F_2) }$ the following equality
			\begin{align*}
				&   \abs{ \iint\, F_1(R_{X_1}( \Theta\,   \delta({X_1,X_2}) X))  \, F_2(R_{X_2}( Y)) \,       
					\langle \Psi_{\omega}|\,
					\phi(X_1 ) \phi(X_2 )  \,\Psi_{\omega} \rangle\, \total^N X_1 \total^N X_2} \\
				&=   \abs{ \iint  \, F_1(R_{X_1}( \Theta\,   \delta({X_1,X_2}) X))  \, F_2(R_{X_2}( Y)) \,      
					\langle \phi(X_1 )\Psi_{\omega}|\,
					\phi(X_2 )  \,\Psi_{\omega} \rangle\, \total^N X_1 \total^N X_2}
				\\
				&\leq C  \norm{\alpha_{\Theta \,  \delta  X}(F_1) \alpha_Y(F_2))}_{\infty}
				\eqend{,}
			\end{align*}
			where we used the Cauchy-Schwarz inequality, the conclusion in \cite[Proof Proposition III.3]{MF}
			and the  supremum-norm $\norm{G }_{\infty}$ of the product of the product of the translated and retracted test-functions $G_{X,Y}:=\alpha_{\Theta \,  \delta  X}(F_1) \alpha_Y(F_2))$ determined by
			\begin{align*}
				\norm{ G_{X,Y} }_{\infty}= \sup_{X_1,X_2\in\mathcal{M}} \vert G_{X,Y}(X_1, X_2) \vert ,
			\end{align*}
			where  we write $X_1,X_2\in\mathcal{M}$ for $X_1,X_2\in\R^N:F(X_{1,2})=0$, where $F$ is the local defining function. Since we choose $G_{X,Y}\in\mathcal{D}^{m}_{0,\rho,\delta} (\mathcal{M})$ the  absolute value of the function $b_\Theta(X,Y,F_1,F_2)$ and derivatives thereof are bounded  such that $b$ is a symbol and therefore the Integral \ref{eq:thm_scalarprod} is well-defined.

		\end{proof}

		\subsection{The Microlocal Spectrum Condition under Deformation Quantization}
		In this section we prove that the micro-local spectrum condition   holds for deformed states, if the undeformed state is Hadamard.  This will be proven up to the second order in the deformation parameter (see the expansion given in Proposition \ref{prop:defpro2}).

		\begin{lemma}\label{lem:2ptfct}
			The deformed two point function defined by the star product in Proposition \ref{prop:defpro2} is,  up to second order in $\Theta$, explicitly presented by 
			\begin{align*}
				\omega^{\Theta}_2(X_1,X_2) &= P^{\Theta}  
				\omega _2(X_1,X_2)
				\\&=\omega _2(X_1,X_2) -i \,\partial_A\partial_{B'} (\Theta^{AB'} ({X_1,X_2})  \omega _2(X_1,X_2))
				\\&\qquad\qquad \qquad\qquad  -
				\nabla_A\,\partial_B  \nabla_{C'}  \partial_{D'}  ( \Theta^{ABC'D'} ({X_1,X_2})  \omega _2(X_1,X_2)),  
			\end{align*} 
			where $P^{\Theta}$ is a  fourth-order  differential operator with smooth coefficients, that depend  on the Poisson bivector, the orthogonal projection and the parallel transport. 
		\end{lemma}
		\begin{proof}
			Using Definition \ref{def:2ptfct}, Proposition \ref{prop:defpro2} and assuming that the functions $(F_1,F_2)\in\mathcal{D}^{m}_{\rho,\delta}(\ma)$ we write the smeared deformed two-point function as
			\begin{align*}
				\omega^{\Theta}_2( \phi(F_1)    \phi(F_2) )&:=\langle \Psi_{\omega}|\,
				\pi_{\omega}(\phi(F_1))\star_\theta  \pi_{\omega}(\phi(F_2)) \,\Psi_{\omega} \rangle\\&= \int  F_1(X_1)\star_\theta  F_2(X_2)\,\omega_2( \phi( X_1 )    \phi( X_2 ) )  \, \total^NX_1 \total^N X_2 \\&=: \int  F_1(X_1)   F_2(X_2)\,\omega^{\Theta}_2( \phi( X_1 )    \phi( X_2 ) )  \, \total^NX_1 \total^N X_2 \\&=\int  F_1(X_1)   F_2(X_2)\,\biggl(\omega _2(X_1,X_2) -i \,\partial_A\partial_{B'} (\Theta^{AB'} ({X_1,X_2})  \omega _2(X_1,X_2))
				\\&-
				\nabla_A\,\partial_B  \nabla_{C'}  \partial_{D'}  ( \Theta^{ABC'D'} ({X_1,X_2})  \omega _2(X_1,X_2))
				\biggr)+\mathcal{O}(\theta^3)\\&=\int  F_1(X_1)   F_2(X_2)\,P^{\Theta}\,\omega _2(X_1,X_2)  +\mathcal{O}(\theta^3),
			\end{align*}
			where in the last lines we inserted the explicit form of the deformed product given in  Proposition \ref{prop:defpro2}   and     we integrated by parts. The smoothness of the differential operator $P$ follows from the smoothness of the Poisson bivector (see Definition \ref{def:poisbivec}), the orthogonal projection (see Lemma \ref{lem:orthproj}) and the parallel transport. 
		\end{proof}

		\begin{theorem}\label{thm:defmsc} 
			Let the state $\omega_2$ obey the microlocal spectrum condition. Then, the deformed state $\omega^{\Theta}$   obeys the microlocal spectrum condition up to second order in $\Theta$, i.e.\ 
			\begin{align*}
				WF(\omega^{\Theta}_2)\subset \mathcal{N}^+\times\mathcal{N}^-,
			\end{align*}
			where the deformed two-point function $\omega^{\Theta}_2$ is given explicitly in Lemma \ref{lem:2ptfct}. 
		\end{theorem}
		\begin{proof}
			By using  the smoothness of the differential operator $P^{\Theta}$ (see Lemma \ref{lem:2ptfct}), the following chain follows from Property \ref{propwf3} of the wavefront set and the fact that $\omega_2$ fulfills the microlocal spectrum condition  
			\begin{align*}
				WF(\omega^{\Theta}_2) =  WF(P^{\Theta}\,\omega_2) \subset WF(\omega_2) \subset \mathcal{N}^+\times\mathcal{N}^-.
			\end{align*}
			
		\end{proof}
		\begin{corollary}
			The deformed two-point function $\omega^{\Theta}_2$ is Hadamard.
		\end{corollary} 
		\begin{proof}
			The fact that the deformed two-point function is Hadamard follows directly from Theorem \ref{thm:defmsc} and Theorem \ref{thm:msceqhc}. 
		\end{proof}
		\begin{theorem} \label{thm:isom}
			Let the deformed states $\omega^{\theta}_2$ and $\omega'^{\theta}_2$ be defined using either two different isomorphic embeddings or/and retractions. Then, the deformed two-point functions $\omega^{\Theta}_2$ and $\omega'^{\Theta}_2$ have the same wavefront set, i.e.
			\begin{align*}
				WF(\omega^{\Theta}_2)=WF(\omega'^{\Theta}_2).
			\end{align*}
		\end{theorem}
		\begin{proof} 
			Due to the smoothness of the embedding or retractions, the deformed two-point functions $\omega^{\Theta}_2$ and $\omega'^{\Theta}_2$ obey the $\mu$SC  and according to Theorem \ref{thm:cinf}, their difference is smooth. Hence, by Property \ref{propwf1} we have
			\begin{align*}
				WF(\omega^{\Theta}_2-\omega'^{\Theta}_2)=0
			\end{align*}
			and Property \ref{propwf2}   the equates the two wavefront sets.
		\end{proof}

		\subsection{Physical Discussion}
		The significance of a state being in Hadamard form, particularly with regard to its two-point functions, cannot be overstated. In a quantum field theory (QFT) within Minkowski space, there exists a preferred state: the vacuum state, which remains invariant under the full Poincaré group. This preferred state naturally entails a preferred choice of the Hilbert space. However, in curved spacetimes, such a preferred choice generally does not exist (except for stationary spacetimes). In fact, different choices in construction typically result in unitarily inequivalent theories. There is no clear criterion for selecting a preferred state in general. Nevertheless, the Hadamard condition on states serves as a prerequisite for a state to possess a finite stress-energy momentum tensor. Adherence to the Hadamard condition significantly restricts the range of unitarily inequivalent Hilbert space constructions within the theory. Therefore, this condition appears to be the appropriate requirement when considering QFT in curved spacetimes. In addition to these arguments, as mentioned earlier, a Hadamard state shares the same singularity structure as the state in Minkowski space, which is a quantum field theoretical compliance with the equivalence principle.
		\par
		In our findings, we established that if the original state is in Hadamard form, then its deformed counterpart also follows suit. This conclusion can be summarized as follows: the singularity pattern of a state in a quantum field theory within noncommutative spacetimes adheres to the same principles as those in Minkowski space, thereby upholding the equivalence principle in that context as well. 
		\par 
		In particular, we can write any  deformed state (where the undeformed is Hadamard, see Definition \ref{def_HadamardFormScalar}) as 
		\begin{align*}
			\omega^{\Theta}_2 = \frac{u}{\sigma}+v \ln{\sigma}+w^{\Theta},
		\end{align*}
		where $u$ and $v$ are purely geometric quantities depending on the manifold $(M,g)$, i.e.\ they are state-independent and therefore neither  depend on the deformation nor the isometric embedding chosen (see Theorem \ref{thm:isom}). The function $w^{\Theta}$ is a smooth function that contains the information of the deformed state and thus the deformation. Thus, the deformed states can as well be understood as "physically realizable" states since they fulfill the principle of local definiteness, see \cite{HaagNarn, haaglocal, verch94}.  Consequently, when assessing the deformed Einstein equations or semi-semiclassical Einstein equations, we can utilize the deformed states to define the (finite) renormalized expectation value of the stress tensor. 
		\section{Conclusion and Outlook}

		In this work,  we provided a specific extension of the Rieffel deformation for all globally hyperbolic spacetimes with a certain Poisson structure. This novel deformation technique is a symbiosis between the Kontsevich deformation \cite{Kont} (since we define a star product for Poisson manifolds), the Rieffel product, and the Fedosov quantization (due to the condition $\nabla \pi=0$). It is tailored such that it can be directly applied to quantum field theories and, in particular, to the central physical object, namely, the two-point function. 
		\\\\
		In addition to the introduction of a quantum-type feature (non-commutativity) to spacetime itself (see Proposition \ref{prop:ncspacetime}), we proved that if a state fulfills the microlocal spectral condition, i.e.\ has a certain sensible singularity structure, then its deformed counterpart will fulfill the $\mu$SC as well. Thus, by the Radzikowksi theorem we conclude that the deformed quasi-free states are  Hadamard. 
		\\\\
		This result sheds light into the  question of whether the behavior at short distances (Hadamard criterion) has a physical foundation due to the break down of the spacetime continuum at scales around the Planck length. In particular, this means that the introduction of a quantum nature (or non-commutativity) to spacetime does not violate the equivalence principle,  from a quantum field theoretical perspective.
		\\\\
		These basic deformation principles can be applied to the various achievements and insights of quantum field theory in curved spacetimes to evaluate its physical validity as being a first-order approximation to a theory of quantum gravity. In addition to introducing a quantum nature to spacetime,  the deformation quantization can be considered a first-order approximation since in the limit $\Theta \rightarrow 0$ we obtain the undeformed, i.e.\ regular,  version of QFT in CST.
		\\\\
		The first place to look for the physical consequences imposed by the deformation quantization are quantum energy inequalities. These are restrictions obtained from QFT towards the domain of violation of classical energy conditions.
		These inequalities are sensitive to the spacetime geometry, e.g. restrict certain exotic spacetimes and effects (such as wormholes, time-machines or warp drives), and take into account the quantum behaviour implied by quantum field theory. 
		Hence, implementing a quantum behavior on spacetime itself may change these inequalities in a non-trivial manner and therewith induce new physical and  measurable effects. 
		\\\\
		Although we defined a deformed product for the case of globally hyperbolic spacetimes, by using the existence of an isometric embedding, no such result  exists for  non-globally hyperbolic manifolds. These cases are in particular interesting when investigating black-holes. Although no isometric embedding will exist in general (e.g.\ in the presence of closed time-like curves), one still can define deformations for QFTs in non-globally hyperbolic spacetimes by using the very idea of how to study these theories \cite{K92}. In principle, one requires the existence of a globally hyperbolic neighborhood $N$ for every point in a given spacetime $\mathcal{M}$. Hence, the  observable algebra of the spacetime $\mathcal{M}$ restricted to $N$ is equal to the algebra one would obtain if one regards  $N$ as a globally hyperbolic spacetime in its own right. Regarding deformation, we can therefore deform the product in the globally hyperbolic neighborhood $N$ by using the embedding formalism and thus define a deformed state. 
		
		\section*{Acknowledgements}
		The author would like to thank R. Verch for various discussions on this topic. Moreover, we extend our gratitude to H. Grosse, in particular on discussions related to the quasi-free deformed states. We would further like to extend our thanks to M. Fröb whose encouragement, sharp remarks and questions shaped many ideas in this paper.  O. M\"uller is gratefully acknowledged for various discussions regarding the embedding formalism. Furthermore, the author expresses his gratitude to R. Ballal, D. Vidal-Cruzprieto and P. Dorau for various discussions that cleared the ideas set in this paper.
		Last but not least, the author thanks Nicolas Boumal for promptly answering questions in regards to optimization techniques. 
		\appendix
		
		\section{Proof of Theorem \ref{thm:defprodorder}}\label{appthm:defprodorder}
		\begin{proof}
			We prove the proposition by using the Taylor expansion of the retractions  in order of the deformation parameter $\theta$,  
			\begin{align*}
				f(R_{Z}( \Theta X))=f(Z)+(\Theta X)^A\,\nabla_Af(Z)+
				\frac{1}{2}  (\Theta X)^A  (\Theta X)^B\nabla_B\nabla_A  \,f(Z)
				+    \mathcal{O}(\theta^3),
			\end{align*}
			We use the following notation
			\begin{align}
				\nabla_A f(Z)=  \pz (\partial_A \bar{f}),  \qquad \qquad(\Theta X)^B\nabla_B\nabla_A  \,f(Z)=\pz((\Theta X)^B\partial_B\,\nabla_A f),
			\end{align}  
			plugging this into the deformation formula we have up to second order in $\theta$, 
			\begin{align*}
				&\left( f \star_\theta g \right)(z) =
				\lim_{\epsilon \to 0} \iint \chi(\epsilon X, \epsilon Y) \, f(R_{Z}( \Theta X))  \, g(R_{Z}( Y)) \,  \,\mathe^{- {\mathi} \,X \cdot\, Y  }  \total^N X \total^N Y\eqend{,}\\ &=
				\lim_{\epsilon \to 0} \iint \chi(\epsilon X, \epsilon Y) \, \biggl(f(Z)+(\Theta X)^A\,\nabla_A f (Z)+  \frac{1}{2} (\Theta X)^A \pz((\Theta X)^B\partial_B\,\nabla_A f) \biggr) \, \\& \qquad \qquad \qquad \times g(R_{Z}( Y)) \,  \,\mathe^{- {\mathi} \,X \cdot\, Y  }  \total^N X \total^N Y\eqend{,} 
			\end{align*} 
			We write  the orthogonal projector as a matrix (see Lemma \ref{lem:orthproj}) 
			and use the orthogonality of this projection rendering the former integral into
			\begin{align*}
				&\left( f \star_\theta g \right)(z)   =f (z)  g(z) -i \Theta^{AB}	 \,\nabla_A f(Z)  \,\nabla_B g(Z)    
				\\  & + \frac{1}{2} \lim_{\epsilon \to 0} \iint \chi(\epsilon X, \epsilon Y)  ( (\Theta X)^A  \pz((\Theta X)^B\partial_B\,\nabla_A f(Z) )) \, g(R_{Z}( Y)) \,  \,\mathe^{- {\mathi} \,X \cdot\, Y  }  \total^N X \total^N Y \\& =f (z)  g(z)-i \Theta^{AB}	 \,\nabla_A f(Z)  \,\nabla_B g(Z)    
				\\  & + \frac{1}{4} \lim_{\epsilon \to 0} \iint \chi(\epsilon X, \epsilon Y)  ( (\Theta X)^A    (\Theta X)^B 
				\partial_B\,\nabla_A f(Z) )) \,
				Y^D Y^F \, 
				\partial_F\,\nabla_D g(Z) ))
				\,  \,\mathe^{- {\mathi} \,X \cdot\, Y  }  \total^N X \total^N Y  \\& =f (z)  g(z)-i \Theta^{AB}	 \,\nabla_A f(Z)  \,\nabla_B g(Z)    
				-\frac{1}{4}
				\,  
				\partial_B\,\nabla_A f(Z)   \,   
				\partial_C\nabla_D g(Z)    (\Theta^{AC}\Theta^{BD}+\Theta^{AD}\Theta^{BC})
				\\& =f (z)  g(z)-i \Theta^{AB}	 \,\nabla_A f(Z)  \,\nabla_B g(Z) - \frac{1}{4} (
				\Theta^{AC}\Theta^{BD}+\Theta^{AD}\Theta^{BC}) \, \partial_A\nabla_B f(Z)\, \partial_C\nabla_D g(Z)    .
			\end{align*}
			Using Lemma \ref{lem:thetproj} and the following relation
			\begin{align*}
				\partial_C\nabla_Dg(Z)=  \partial_C \partial_D\bar{g}(Z)+(\partial_C P_D^{\,\,E})\partial_E \bar{g}(Z)\eqend{,} 
			\end{align*}
			we see that
			\begin{align*}
				\Theta^{AD}\Theta^{BC}\partial_C\nabla_D g(Z)  = \Theta^{AC}\Theta^{BD}\partial_C\nabla_D g(Z).
			\end{align*}
			Therefore, the deformed product of two functions is given up to second order in the deformation parameter as, 
			\begin{align*}
				\left( f \star_\theta g \right)(z)   &=f    g -i \Theta^{AB}	 \,\partial_A f  \,\partial_B g  - \frac{1}{4} (
				\Theta^{AC}\Theta^{BD}+\Theta^{AD}\Theta^{BC})
				\partial_A\nabla_B f \, \partial_C\nabla_D g \\
				&=f    g -i \Theta^{AB}	 \,\partial_A f  \,\partial_B g  - \frac{1}{2} 
				\Theta^{AC}\Theta^{BD} 
				\partial_A\nabla_B f \, \partial_C\nabla_D g .
			\end{align*} 
			Next, we prove associativity where we first consider 
			\begin{align*}
				&   ((f\star_{\theta}g) \star_{\theta}h ) \,  = (F \star_{\theta}h) \,  \\&\,\\ & =F    h -i \Theta^{AB}	 \,\partial_A F  \,\partial_B h  - \frac{1}{2}  
				\Theta^{AC}\Theta^{BD} \,
				\partial_A\nabla_B F \, \partial_C\nabla_D h\\&\,\\ & 
				=(f    g -i \Theta^{AB}	 \,\partial_A f  \,\partial_B g  -  \frac{1}{2}  
				\Theta^{AC}\Theta^{BD} 
				\partial_A\nabla_B f \, \partial_C\nabla_D g    )    h \\&\,\\ & -i \Theta^{AB}	 \,\partial_A (f    g -i \Theta^{CD}	 \,\partial_C f  \,\partial_D g)  \,\partial_B h   \\&\,\\ & -  \frac{1}{2}  
				\Theta^{AC}\Theta^{BD} 
				\partial_A\nabla_B (fg) \, \partial_C\nabla_D h\\&\,\\ & 
				=(f    g -i \Theta^{AB}	 \,\partial_A f  \,\partial_B g  -  \frac{1}{2}  
				\Theta^{AC}\Theta^{BD} \,
				\partial_A\nabla_B f \, \partial_C\nabla_D g    )    h \\&\,\\ & -i \Theta^{AB} (\partial_Af   	 \, g+f  \,\partial_A   g -i \partial_A\Theta^{CD}	 \,\partial_C f  \,\partial_D g-i \Theta^{CD}	 \,\partial_A\partial_C f  \,\partial_D g-i \Theta^{CD}	 \,\partial_C f  \,\partial_A\partial_D g)  \,\partial_B h   \\&\,\\ & -  \frac{1}{2}  
				\Theta^{AC}\Theta^{BD} 
				( \partial_A\nabla_B f\,g+\partial_Af\,
				\partial_B g+
				\partial_B f\,\partial_Ag+f\, \partial_A\nabla_B g) \, \partial_C\nabla_D h\eqend{,} 
			\end{align*} 
			Next, we consider 
			\begin{align*}
				&   (f\star( g \star_{\theta}h )) \,  = (f \star_{\theta}G) \, 
				\\&\,\\ &= 
				f    G -i \Theta^{AB}	 \,\partial_A f  \,\partial_B G  -  \frac{1}{2}  
				\Theta^{AC}\Theta^{BD} 
				\partial_A\nabla_B f \, \partial_C\nabla_D G     \\&\,\\ &= 
				f   (g    h -i \Theta^{AB}	 \,\partial_A g  \,\partial_B h -  \frac{1}{2}  
				\Theta^{AC}\Theta^{BD} 
				\partial_A\nabla_B g \, \partial_C\nabla_D h )  
				\\&\,\\ &-i \Theta^{AB}	 \,\partial_A f  \,\partial_B (g    h -i \Theta^{CD}	 \,\partial_C g  \,\partial_D h  )  
				\\&\,\\ & -  
				\Theta^{AC}\Theta^{BD} 
				\partial_A\nabla_B f \, \partial_C\nabla_D (g    h  )     \\&\,\\ &= 
				f   (g    h -i \Theta^{AB}	 \,\partial_A g  \,\partial_B h -  \frac{1}{2}  
				\Theta^{AC}\Theta^{BD} 
				\partial_A\nabla_B g \, \partial_C\nabla_D h )  
				\\&\,\\ &-i \Theta^{AB}	 \,\partial_A f  \, (\partial_B g  \,  h+g   \,\partial_B h -i \partial_B \Theta^{CD}	 \,\partial_C g  \,\partial_D h -i \Theta^{CD}	 \,\partial_B \partial_C g  \,\partial_D h -i \Theta^{CD}	 \,\partial_C g  \,\partial_B \partial_D h  )  
				\\&\,\\ & -  \frac{1}{2}  
				\Theta^{AC}\Theta^{BD} 
				\partial_A\nabla_B f \,  (\partial_C\nabla_Dg  \,  h +\partial_Cg  \,  \partial_Dh +\partial_Dg  \, \partial_C h +g  \, \partial_C\nabla_D h  )    
			\end{align*} 
			This reduces to 
			\begin{align*}
				&  ( - \Theta^{AD}   \partial_A\Theta^{CB}	 \,\partial_C f  \,\partial_B g \,\partial_D h  - \Theta^{AB} \Theta^{CD}	 \,\partial_A\partial_C f  \,\partial_D g \,\partial_B h )   -  
				\Theta^{AC}\Theta^{BD} 
				\partial_Af\,
				\partial_B g   \, \partial_C\nabla_D h
				\\&\,\\ &=
				( -  \Theta^{CA} \,   \partial_A \Theta^{BD}		 \,\partial_C f  \,\partial_B g  \,\partial_D h    -  \Theta^{AB}	\Theta^{CD}	 \,\partial_A f  \, \partial_C g  \,\partial_B \partial_D h  )  
				- 
				\Theta^{AC}\Theta^{BD} 
				\partial_A\nabla_B f \,  (\partial_Cg  \,  \partial_Dh  )    \eqend{,} 
			\end{align*}
			which further reduces to 
			\begin{align*}
				\Theta^{AB}   \partial_A\Theta^{DC}	   - 
				\Theta^{CA}\Theta^{BE}\partial_A P^{D}_{E}
				+ 
				\Theta^{BE}\Theta^{DA}
				\partial_A P^{C}_{E}   =0,  
			\end{align*}
			or using Lemma \ref{lem:thetproj} can be written as 
			\begin{align*}
				\Theta^{AB}   \partial_A\Theta^{DC}	   - 
				\Theta^{CE}\Theta^{BA}\partial_A P^{D}_{E}
				+ 
				\Theta^{BA}\Theta^{DE}
				\partial_A P^{C}_{E}   =0.
			\end{align*}
			Using Lemma \ref{lemmpi0}, this equality is equivalent to
			\begin{align*}
				\Theta^{AB}   \nabla_A\Theta^{DC}	 =0,
			\end{align*}
			hence concluding the proof of associativity to second order.

		\end{proof}
		\section{Proof of Proposition \ref{prop:defpro2}}
		\begin{proof}
			We prove the proposition by using the Taylor expansion of the retractions  in order of the deformation parameter $\theta$,  
			\begin{align*}
				f(R_{Z_1}( \Theta X))=f(Z_1)+(\Theta\,\delta({Z_1,Z_2})) X)^A\,\nabla_Af(Z_1)+
				\frac{1}{2}  (\Theta \,\delta({Z_1,Z_2}))X)^A  \nabla_{\Theta \,\delta({Z_1,Z_2})) X} \nabla_A  \,f(Z_1) \eqend{.} 
			\end{align*}
			We use the following notation
			\begin{align}
				\nabla_A f(Z_1)=  \pz (\partial_A \bar{f}(Z_1)),  \qquad \qquad \nabla_{Z} \nabla_A  \,f(Z_1)=\pza(Z^B\partial_B\,\nabla_A f(Z_1))\eqend{,} 
			\end{align}  
			plugging this into the deformation formula we have up to second order in $\theta$, 
			\begin{align*}
				& f (z_1)\star_\theta g (z_2) =
				\lim_{\epsilon \to 0} \iint \chi(\epsilon X, \epsilon Y) \, f(R_{Z_1}( \Theta \,\delta(Z_1,Z_2) \,X))  \, g(R_{Z_2}( Y)) \,  \,\mathe^{- {\mathi} \,X \cdot\, Y  }  \total^N X \total^N Y \\ &=
				\lim_{\epsilon \to 0} \iint \chi(\epsilon X, \epsilon Y) \, \biggl(f(Z_1)+(\Theta \,\delta({Z_1,Z_2}) X)^A\,\nabla_A f (Z_1)\\ &+  \frac{1}{2} (\Theta \,\delta({Z_1,Z_2}) X)^A \pza((\Theta\,\delta({Z_1,Z_2}) X)^B\partial_B\,\partial_A f(Z_1)) \biggr) \,   g(R_{Z_2}( Y)) \,  \,\mathe^{- {\mathi} \,X \cdot\, Y  }  \total^N X \total^N Y\\ &=f (z_1)  g(z_2) + \Theta^{AB}\,	\delta^{\,\,\,C'}_{B}({Z_1,Z_2}) \,\nabla_A f(Z_1)  \lim_{\epsilon \to 0} \iint \chi(\epsilon X, \epsilon Y) \,X_{C'}    \, g(R_{Z_2}( Y)) \,  \,\mathe^{- {\mathi} \,X \cdot\, Y  }  \total^N X \total^N 
				Y
				\\  & + \frac{1}{2} \lim_{\epsilon \to 0} \iint \chi(\epsilon X, \epsilon Y)  ( (\Theta \,\delta({Z_1,Z_2}) X)^A  \pza((\Theta\,\delta({Z_1,Z_2}) X)^B\partial_B\,\nabla_A f(Z_1) )) \\  &\qquad\qquad \times \, g(R_{Z_2}( Y)) \,  \,\mathe^{- {\mathi} \,X \cdot\, Y  }  \total^N X \total^N Y
				\\ &=f (z_1)  g(z_2)-i \Theta^{AB}\,	\delta^{\,\,\,C'}_{B}({Z_1,Z_2}) \,\nabla_A f(Z_1)   \,\nabla_{C'} g(Z_2)   
				\\  & + \frac{1}{4} \lim_{\epsilon \to 0} \iint \chi(\epsilon X, \epsilon Y)  ( (\Theta \,\delta({Z_1,Z_2}) X)^A  \pza((\Theta\,\delta({Z_1,Z_2}) X)^B\partial_B\,\nabla_A f(Z_1) )) \\  &\qquad\qquad \times \, 
				Y^{A'}Y^{B'}\, P^{D'}_{\,\,A'}(Z_2)\, \partial_{B'}\nabla_{D'}g(Z_2)
				\,\mathe^{- {\mathi} \,X \cdot\, Y  }  \total^N X \total^N Y
				\\ &=f (z_1)  g(z_2)-i \Theta^{AB}\,	\delta^{\,\,\,C'}_{B}({Z_1,Z_2}) \,\nabla_A f(Z_1)   \,\nabla_{C'} g(Z_2)   
				\\  & + \frac{1}{4} 
				\Theta^{AC}\,\delta_C^{\,\,\,C'}(Z_1,Z_2) \,   \Theta^{BD}\,\delta_D^{\,\,\,D'}(Z_1,Z_2) \partial_B\,\nabla_A f(Z_1)\,   P^{D'}_{\,\,A'}(Z_2)\partial_{B'}\nabla_{D'}g(Z_2) \\  &\qquad\qquad \times
				\lim_{\epsilon \to 0} \iint \chi(\epsilon X, \epsilon Y)   X_{C'}X_{D'}
				Y^{A'}Y^{B'}  
				\,\mathe^{- {\mathi} \,X \cdot\, Y  }  \total^N X \total^N Y   \\ &=f (z_1)  g(z_2)-i \Theta^{AB}\,	\delta^{\,\,\,C'}_{B}({Z_1,Z_2}) \,\nabla_A f(Z_1)   \,\nabla_{C'} g(Z_2)   
				\\  & - \frac{1}{4} 
				\Theta^{AC}\,\delta_C^{\,\,\,C'}(Z_1,Z_2) \,   \Theta^{BD}\,\delta_D^{\,\,\,D'}(Z_1,Z_2) \partial_B\,\nabla_A f(Z_1)\,   P^{E'}_{\,\,A'}(Z_2)\partial_{B'}\nabla_{E'}g(Z_2) \\  &\qquad\qquad \times
				\left(\delta^{\,\,\,A'}_{D'}\delta^{\,\,\,B'}_{C'}+\delta^{\,\,\,B'}_{D'}\delta^{\,\,\,A'}_{C'} \right)
				\\ &=f (z_1)  g(z_2)-i \Theta^{AB}\,	\delta^{C'}_{B}({Z_1,Z_2}) \,\nabla_A f(Z_1)   \,\nabla_{C'} g(Z_2)   
				\\  & - \frac{1}{4}  
				\Theta^{AC}\,\delta_C^{\,\,\,C'}(Z_1,Z_2) \,   \Theta^{BD}\,\delta_D^{\,\,\,A'}(Z_1,Z_2) \partial_B\,\nabla_A f(Z_1)\,   P^{E'}_{\,\,A'}(Z_2)\partial_{C'}\nabla_{E'}g(Z_2)
				\\  & - \frac{1}{4}  
				\Theta^{AC}\,\delta_C^{\,\,\,A'}(Z_1,Z_2) \,   \Theta^{BD}\,\delta_D^{\,\,\,C'}(Z_1,Z_2) \partial_B\,\nabla_A f(Z_1)\,   P^{E'}_{\,\,A'}(Z_2)\partial_{C'}\nabla_{E'}g(Z_2)  \\ &=f (z_1)  g(z_2)-i \Theta^{AB}\,	\delta^{\,\,\,C'}_{B}({Z_1,Z_2}) \,\nabla_A f(Z_1)   \,\nabla_{C'} g(Z_2)   
				\\  & -  \frac{1}{4}  \left( \Theta^{AC}\,  \Theta^{BD}+ \Theta^{AD}\, \Theta^{BC}  \right)
				\delta_D^{\,\,\,A'}(Z_1,Z_2)\,\delta_C^{\,\,\,C'}(Z_1,Z_2)  \,   \partial_B\,\nabla_A f(Z_1)\,   P^{E'}_{\,\,A'}(Z_2)\partial_{C'}\nabla_{E'}g(Z_2) 
				\eqend{.}
			\end{align*}  
			Next, we use  Lemma \ref{lem:thetproj} and Lemma \ref{lem:thetproj'} and simplify the terms to 
			\begin{align*}
				& f (z_1)\star_\theta g (z_2) =  f (z_1)  g(z_2)-i \Theta^{AC'} ({Z_1,Z_2}) \,\partial_A f(Z_1)   \,\partial_{C'} g(Z_2)   
				\\  &-  \frac{1}{4}  \left( \Theta^{AC'}(Z_1,Z_2) \,  \Theta^{BD'}(Z_1,Z_2) + \Theta^{AD'}(Z_1,Z_2) \, \Theta^{BC'} (Z_1,Z_2)  \right)  \,   \partial_A\,\nabla_B f(Z_1)\,  (Z_2)\partial_{C'}\nabla_{D'}g(Z_2)  \eqend{.} 
			\end{align*}
			
		\end{proof}   
		\section{Proof of Proposition \ref{prop:defpro2}}\label{approp:defpro2}
		\begin{proof}
			We prove the proposition by using the Taylor expansion of the retractions  in order of the deformation parameter $\theta$,  
			\begin{align*}
				f(R_{Z_1}( \Theta X))=f(Z_1)+(\Theta\,\delta({Z_1,Z_2})) X)^A\,\nabla_Af(Z_1)+
				\frac{1}{2}  (\Theta \,\delta({Z_1,Z_2}))X)^A  \nabla_{\Theta \,\delta({Z_1,Z_2})) X} \nabla_A  \,f(Z_1) \eqend{,} 
			\end{align*}
			We use the following notation
			\begin{align}
				\nabla_A f(Z_1)=  \pz (\partial_A \bar{f}(Z_1)),  \qquad \qquad \nabla_{Z} \nabla_A  \,f(Z_1)=\pza(Z^B\partial_B\,\nabla_A f(Z_1))\eqend{,} 
			\end{align}  
			plugging this into the deformation formula we have up to second order in $\theta$, 
			\begin{align*}
				& f (z_1)\star_\theta g (z_2) =
				\lim_{\epsilon \to 0} \iint \chi(\epsilon X, \epsilon Y) \, f(R_{Z_1}( \Theta \,\delta(Z_1,Z_2) \,X))  \, g(R_{Z_2}( Y)) \,  \,\mathe^{- {\mathi} \,X \cdot\, Y  }  \total^N X \total^N Y \\ &=
				\lim_{\epsilon \to 0} \iint \chi(\epsilon X, \epsilon Y) \, \biggl(f(Z_1)+(\Theta \,\delta({Z_1,Z_2}) X)^A\,\nabla_A f (Z_1)\\ &+  \frac{1}{2} (\Theta \,\delta({Z_1,Z_2}) X)^A \pza((\Theta\,\delta({Z_1,Z_2}) X)^B\partial_B\,\partial_A f(Z_1)) \biggr) \,   g(R_{Z_2}( Y)) \,  \,\mathe^{- {\mathi} \,X \cdot\, Y  }  \total^N X \total^N Y 
				\\ &=f (z_1)  g(z_2)-i \Theta^{AB}\,	\delta^{\,\,\,C'}_{B}({Z_1,Z_2}) \,\nabla_A f(Z_1)   \,\nabla_{C'} g(Z_2)   
				\\  & + \frac{1}{4} \lim_{\epsilon \to 0} \iint \chi(\epsilon X, \epsilon Y)  ( (\Theta \,\delta({Z_1,Z_2}) X)^A  \pza((\Theta\,\delta({Z_1,Z_2}) X)^B\partial_B\,\nabla_A f(Z_1) )) \\  &\qquad\qquad \times \, 
				Y^{A'}Y^{B'}\, P^{D'}_{\,\,A'}(Z_2)\, \partial_{B'}\nabla_{D'}g(Z_2)
				\,\mathe^{- {\mathi} \,X \cdot\, Y  }  \total^N X \total^N Y
				\\ &=f (z_1)  g(z_2)-i \Theta^{AB}\,	\delta^{\,\,\,C'}_{B}({Z_1,Z_2}) \,\nabla_A f(Z_1)   \,\nabla_{C'} g(Z_2)   
				\\  & - \frac{1}{4} 
				\Theta^{AC}\,\delta_C^{\,\,\,C'}(Z_1,Z_2) \,   \Theta^{BD}\,\delta_D^{\,\,\,D'}(Z_1,Z_2) \partial_B\,\nabla_A f(Z_1)\,   P^{E'}_{\,\,A'}(Z_2)\partial_{B'}\nabla_{E'}g(Z_2) \\  &\qquad\qquad \times
				\left(\delta^{\,\,\,A'}_{D'}\delta^{\,\,\,B'}_{C'}+\delta^{\,\,\,B'}_{D'}\delta^{\,\,\,A'}_{C'} \right)
				\\ &  =f (z_1)  g(z_2)-i \Theta^{AB}\,	\delta^{\,\,\,C'}_{B}({Z_1,Z_2}) \,\nabla_A f(Z_1)   \,\nabla_{C'} g(Z_2)   
				\\  & -  \frac{1}{4}  \left( \Theta^{AC}\,  \Theta^{BD}+ \Theta^{AD}\, \Theta^{BC}  \right)
				\delta_D^{\,\,\,A'}(Z_1,Z_2)\,\delta_C^{\,\,\,C'}(Z_1,Z_2)  \,   \partial_B\,\nabla_A f(Z_1)\,   P^{E'}_{\,\,A'}(Z_2)\partial_{C'}\nabla_{E'}g(Z_2) 
				\eqend{.}
			\end{align*}  
			Next, we use  Lemma \ref{lem:thetproj} and Lemma \ref{lem:thetproj'} and simplify the terms to 
			\begin{align*}
				& f (z_1)\star_\theta g (z_2) =  f (z_1)  g(z_2)-i \Theta^{AC'} ({Z_1,Z_2}) \,\partial_A f(Z_1)   \,\partial_{C'} g(Z_2)   
				\\  &-  \frac{1}{4}  \left( \Theta^{AC'}(Z_1,Z_2) \,  \Theta^{BD'}(Z_1,Z_2) + \Theta^{AD'}(Z_1,Z_2) \, \Theta^{BC'} (Z_1,Z_2)  \right)  \,   \partial_A\,\nabla_B f(Z_1)\,  (Z_2)\partial_{C'}\nabla_{D'}g(Z_2)  \eqend{.} 
			\end{align*}
			
		\end{proof}   
		
		\bibliographystyle{alpha}
		\bibliography{allliterature1b.bib}

	\end{document}